\documentclass[pra,aps,onecolumn,twoside,superscriptaddress]{revtex4}

\usepackage{amsmath,amsfonts,amssymb,caption,color,epsfig,graphics,graphicx,hyperref,latexsym,mathrsfs,revsymb,theorem,url,verbatim,epstopdf,multirow,array}
\usepackage[matrix,frame,arrow]{xypic}
\usepackage[qm]{qcircuit}

\hypersetup{colorlinks,linkcolor={blue},citecolor={red},urlcolor={blue}}

\newtheorem{definition}{Definition}
\newtheorem{proposition}[definition]{Proposition}
\newtheorem{lemma}[definition]{Lemma}

\newtheorem{theorem}[definition]{Theorem}
\newtheorem{corollary}[definition]{Corollary}
\newtheorem{conjecture}[definition]{Conjecture}

\newtheorem{remark}[definition]{Remark}
\newtheorem{example}[definition]{Example}
\newtheorem{question}[definition]{Question}

\def\bcj{\begin{conjecture}}
\def\ecj{\end{conjecture}}
\def\bcr{\begin{corollary}}
\def\ecr{\end{corollary}}
\def\bd{\begin{definition}}
\def\ed{\end{definition}}
\def\bea{\begin{eqnarray}}
\def\eea{\end{eqnarray}}
\def\bem{\begin{enumerate}}
\def\eem{\end{enumerate}}
\def\bex{\begin{example}}
\def\eex{\end{example}}
\def\bim{\begin{itemize}}
\def\eim{\end{itemize}}
\def\bl{\begin{lemma}}
\def\el{\end{lemma}}
\def\bma{\begin{bmatrix}}
\def\ema{\end{bmatrix}}
\def\bpf{\begin{proof}}
\def\epf{\end{proof}}
\def\bpp{\begin{proposition}}
\def\epp{\end{proposition}}
\def\bqu{\begin{question}}
\def\equ{\end{question}}
\def\br{\begin{remark}}
\def\er{\end{remark}}
\def\bt{\begin{theorem}}
\def\et{\end{theorem}}

\def\squareforqed{\hbox{\rlap{$\sqcap$}$\sqcup$}}
\def\qed{\ifmmode\squareforqed\else{\unskip\nobreak\hfil
\penalty50\hskip1em\null\nobreak\hfil\squareforqed
\parfillskip=0pt\finalhyphendemerits=0\endgraf}\fi}
\def\endenv{\ifmmode\;\else{\unskip\nobreak\hfil
\penalty50\hskip1em\null\nobreak\hfil\;
\parfillskip=0pt\finalhyphendemerits=0\endgraf}\fi}
\newenvironment{proof}{\noindent \textbf{{Proof.~} }}{\qed}
\def\Dbar{\leavevmode\lower.6ex\hbox to 0pt
{\hskip-.23ex\accent"16\hss}D}
\makeatletter
\def\url@leostyle{%
  \@ifundefined{selectfont}{\def\UrlFont{\sf}}{\def\UrlFont{\small\ttfamily}}}
\makeatother
\def\bcj{\begin{conjecture}}
\def\ecj{\end{conjecture}}
\def\bcr{\begin{corollary}}
\def\ecr{\end{corollary}}
\def\bd{\begin{definition}}
\def\ed{\end{definition}}
\def\bea{\begin{eqnarray}}
\def\eea{\end{eqnarray}}
\def\bem{\begin{enumerate}}
\def\eem{\end{enumerate}}
\def\bex{\begin{example}}
\def\eex{\end{example}}
\def\bim{\begin{itemize}}
\def\eim{\end{itemize}}
\def\bl{\begin{lemma}}
\def\el{\end{lemma}}
\def\bpf{\begin{proof}}
\def\epf{\end{proof}}
\def\bpp{\begin{proposition}}
\def\epp{\end{proposition}}
\def\bqu{\begin{question}}
\def\equ{\end{question}}
\def\br{\begin{remark}}
\def\er{\end{remark}}
\def\bt{\begin{theorem}}
\def\et{\end{theorem}}

\def\btb{\begin{tabular}}
\def\etb{\end{tabular}}

\newcommand{\etal}{{\sl et~al.}}

\newcommand{\nc}{\newcommand}

\def\a{\alpha}
\def\b{\beta}
\def\g{\gamma}
\def\d{\delta}

\def\t{\theta}

\def\p{\pi}

\def\s{\sigma}

\def\ph{\varphi}

 \nc{\bbA}{\mathbb{A}} \nc{\bbB}{\mathbb{B}} \nc{\bbC}{\mathbb{C}}
 \nc{\bbD}{\mathbb{D}} \nc{\bbE}{\mathbb{E}} \nc{\bbF}{\mathbb{F}}
 \nc{\bbG}{\mathbb{G}} \nc{\bbH}{\mathbb{H}} \nc{\bbI}{\mathbb{I}}
 \nc{\bbJ}{\mathbb{J}} \nc{\bbK}{\mathbb{K}} \nc{\bbL}{\mathbb{L}}
 \nc{\bbM}{\mathbb{M}} \nc{\bbN}{\mathbb{N}} \nc{\bbO}{\mathbb{O}}
 \nc{\bbP}{\mathbb{P}} \nc{\bbQ}{\mathbb{Q}} \nc{\bbR}{\mathbb{R}}
 \nc{\bbS}{\mathbb{S}} \nc{\bbT}{\mathbb{T}} \nc{\bbU}{\mathbb{U}}
 \nc{\bbV}{\mathbb{V}} \nc{\bbW}{\mathbb{W}} \nc{\bbX}{\mathbb{X}}
 \nc{\bbZ}{\mathbb{Z}}

 \nc{\bA}{{\bf A}} \nc{\bB}{{\bf B}} \nc{\bC}{{\bf C}}
 \nc{\bD}{{\bf D}} \nc{\bE}{{\bf E}} \nc{\bF}{{\bf F}}
 \nc{\bG}{{\bf G}} \nc{\bH}{{\bf H}} \nc{\bI}{{\bf I}}
 \nc{\bJ}{{\bf J}} \nc{\bK}{{\bf K}} \nc{\bL}{{\bf L}}
 \nc{\bM}{{\bf M}} \nc{\bN}{{\bf N}} \nc{\bO}{{\bf O}}
 \nc{\bP}{{\bf P}} \nc{\bQ}{{\bf Q}} \nc{\bR}{{\bf R}}
 \nc{\bS}{{\bf S}} \nc{\bT}{{\bf T}} \nc{\bU}{{\bf U}}
 \nc{\bV}{{\bf V}} \nc{\bW}{{\bf W}} \nc{\bX}{{\bf X}}
 \nc{\bZ}{{\bf Z}}

\nc{\hA}{{\hat{A}}} \nc{\hB}{{\hat{B}}} \nc{\hC}{{\hat{C}}}
\nc{\hD}{{\hat{D}}} \nc{\hE}{{\hat{E}}} \nc{\hF}{{\hat{F}}}
\nc{\hG}{{\hat{G}}} \nc{\hH}{{\hat{H}}} \nc{\hI}{{\hat{I}}}
\nc{\hJ}{{\hat{J}}} \nc{\hK}{{\hat{K}}} \nc{\hL}{{\hat{L}}}
\nc{\hM}{{\hat{M}}} \nc{\hN}{{\hat{N}}} \nc{\hO}{{\hat{O}}}
\nc{\hP}{{\hat{P}}} \nc{\hR}{{\hat{R}}} \nc{\hS}{{\hat{S}}}
\nc{\hT}{{\hat{T}}} \nc{\hU}{{\hat{U}}} \nc{\hV}{{\hat{V}}}
\nc{\hW}{{\hat{W}}} \nc{\hX}{{\hat{X}}} \nc{\hZ}{{\hat{Z}}}

\nc{\hn}{{\hat{n}}}

\def\diag{\mathop{\rm diag}}

\def\ghz{\mathop{\rm GHZ}}

\def\lin{\mathop{\rm span}}

\def\w{\mathop{\rm W}}

\def\bigox{\bigotimes}

\def\lra{\leftrightarrow}

\def\ox{\otimes}

\def\beq{\begin{equation}}
\def\eeq{\end{equation}}
\def\bal{\begin{aligned}}
\def\eal{\end{aligned}}

\newcommand{\ket}[1]{|#1\rangle}
\newcommand{\proj}[1]{| #1\rangle\!\langle #1 |}
\newcommand{\ketbra}[2]{|#1\rangle\!\langle#2|}

\newcommand{\abs}[1]{|#1|}

\def\Dbar{\leavevmode\lower.6ex\hbox to 0pt
{\hskip-.23ex\accent"16\hss}D}

\begin{document}

\title{Classification of Schmidt-rank-two multipartite unitary gates by singular number}

\date{\today}


\author{Yi Shen}
\email[]{yishen@jiangnan.edu.cn}
\affiliation{School of Science, Jiangnan University, Wuxi, Jiangsu 214122, China}

\author{Lin Chen}
\email[]{linchen@buaa.edu.cn (corresponding author)}
\affiliation{School of Mathematical Sciences, Beihang University, Beijing 100191, China}
\affiliation{International Research Institute for Multidisciplinary Science, Beihang University, Beijing 100191, China}

\author{Li Yu}
\email[]{yuli@hznu.edu.cn}
\affiliation{School of Physics, Hangzhou Normal University, Hangzhou, Zhejiang 311121, China}

\begin{abstract}
The multipartite unitary gates are called genuine if they are not product unitary operators across any bipartition. We mainly investigate the classification of genuine multipartite unitary gates of Schmidt rank two, by focusing on the multiqubit scenario.
For genuine multipartite (excluding bipartite) unitary gates of Schmidt rank two, there is an essential fact that their Schmidt decompositions are unique. Based on this fact, we propose a key notion named as \emph{singular number} to classify the unitary gates concerned. The singular number is defined as the number of local singular operators in the Schmidt decomposition. We then determine the accurate range of singular number. For each singular number, we formulate the parametric Schmidt decompositions of genuine multiqubit unitary gates under local equivalence. 
Finally, we extend the study to three-qubit diagonal unitary gates due to the close relation between diagonal unitary gates and Schmidt-rank-two unitaries. We start with discussing two typical examples of Schmidt rank two, one of which is a fundamental three-qubit unitary gate, i.e., the CCZ gate. Then we characterize the diagonal unitary gates of Schmidt rank greater than two. We show that a three-qubit diagonal unitary gate has Schmidt rank at most three, and present a necessary and sufficient condition for such a unitary gate of Schmidt rank three. This completes the characterization of all genuine three-qubit diagonal unitary gates.
\end{abstract}

\maketitle

\tableofcontents

\section{Introduction}
\label{sec:int}

The unitary evolution realized by unitary operations is a fundamental dynamical process for a quantum system, and has been regarded as valuable physical resource \cite{Nielsen03}. Hence implementing multipartite unitary operations is a key task in quantum information processing. The unitary operation is also known as the unitary gate in a quantum circuit. In particular, multiqubit unitary gates play an essential role in both theory \cite{cy13,cy14,gaterecover2018} and experiment \cite{mcug2003,MGbenchmark2012,MGatoms2019}. Multipartite unitary gates are basically divided into local and nonlocal ones. Specifically, a multipartite unitary gate is called local when it is the tensor product of unitary operators locally acting on subsystems, i.e., it has Schmidt rank one by the operator Schmidt decomposition. Otherwise, it is called nonlocal \cite{cy14}. It is known that local unitary gates can be implemented only by local operations and classical communication (LOCC) with probability one, while nonlocal unitary gates cannot be realized in this way,  even if the probability is allowed to be close to zero \cite{pv98}. In this paper, we mainly investigate a kind of nonlocal unitary gates whose Schmidt rank (SR) is two.

Nonlocal unitary gates play a more powerful role than local unitaries in quantum computing \cite{dqc2007}, cryptography \cite{nlocalcrypt2011}, and so on, in virtue of the vital property that they can create quantum entanglement between distributed parties \cite{ejp00}. The entangling power \cite{cy16,cy16b} which quantifies the maximum output entanglement of a nonlocal unitary gate is an effective measure to evaluate how useful it is for quantum information processing. The simplest type of nonlocal unitary is the \emph{controlled unitary} gates. A bipartite unitary gate $U_{AB}$ is said controlled from system $A$ if it is in the form as $U_{AB}=\sum_{j}^m P_j\ox V_j$, where $P_j$'s are orthogonal projectors on system $A$ and $V_j$'s are unitaries on system $B$. The controlled $U_{AB}$ can be implemented by a simple nonlocal protocol \cite{ygc10} using a maximally entangled state of Schmidt rank $m$. In this sense the implementation of controlled unitaries is operational. Thus, decomposing the complicated unitary gates into the product of controlled unitary gates \cite{cy15} is an instructive way to implement general nonlocal unitaries. Moreover, controlled unitary gates are indispensable for quantum circuits of various uses. For example, the controlled NOT (CNOT) gate is essential to construct the universal two-qubit gate used in quantum computing \cite{Barenco95}. It has also been shown that controlled unitary gates are instrumental to generate multiqubit graph states for one-way quantum computing \cite{br01}, and to construct the mutually unbiased bases (MUBs) \cite{wpz11}.

The research on unitary gates of Schmidt rank two (SR-2) is the first step to study nonlocal unitaries, and becomes the foundation to further study nonlocal unitaries of SR greater than two \cite{cy14,cy14ap}. In addition, there are several widely used nonlocal unitary gates whose SR is two, e.g., the two-qubit CNOT gate and the three-qubit Toffoli gate \cite{Toffoli12}. Hence it is necessary to deeply understand multipartite unitary gates of SR-2 both theoretically and experimentally. There is a fundamental result from \cite{cy13} on the unitaries of SR-2, which builds a close connection to the above-mentioned controlled unitary gates. It states that every nonlocal unitary gate of SR-2 is locally equivalent to i) a fully controlled unitary (controlled from every party of the system), and ii) to a diagonal unitary \cite[Theorem 1]{cy13}. The former implies that one can implement any multipartite unitary gate of SR-2 by implementing some fully controlled unitary gate assisted with a sequence of local unitary gates. The latter reveals an essential structure for multipartite unitary gates of SR-2 under local equivalence.

Furthermore, it should be more useful to obtain specific Schmidt decompositions of SR-2 unitaries which are essentially diagonal unitaries under local equivalence. 
We note that some approximation algorithms to find the Schmidt decomposition under local unitary (LU) equivalence is in fact given in the approach of searching the local invertible operators to transform a three-qubit pure state in the GHZ class into the GHZ state \cite{aacj00}. The process of finding the Schmidt decomposition is equivalent to finding the local invertible operators to obtain a ``standard'' unitary in Schmidt form. These local invertible operators can actually be assumed to be local unitary operators, according to \cite[Theorem 7]{cy14}.
By comparing the diagonal unitary forms from two unitary gates of SR-2, we can easily determine whether the two gates are locally equivalent. Thus, to figure out the parametric Schmidt decompositions under local equivalence is also related to the classification of unitary operators.  
For two-qubit unitary gates, there is an essential characterization under LU equivalence. That is, any two-qubit unitary gate $W_{AB}$ is LU equivalent to the canonic form $\tilde{U}_{AB}$ \cite{kc01}:
\begin{eqnarray}
\label{eq:twoqubit} 
\tilde{U}_{AB}:=\sum^3_{j=0} c_j\s_j \ox \s_j,
\end{eqnarray}
where $\s_j$'s are the Pauli matrices. According to the discussion above, Eq. \eqref{eq:twoqubit} provides a complete classification of all two-qubit unitary gates under LU equivalence. The classification of nonlocal unitary gates has important practical significance \cite{dvc2002}. For example, since local unitary transformations do not change the entanglement, the nonlocal unitary gates in the same equivalence class share the same entangling power. Thus, in order to obtain the entangling power of two-qubit unitary gates it suffices to investigate that of the canonic $\tilde{U}_{AB}$ given by Eq. \eqref{eq:twoqubit} \cite{ylep2018}.
Nevertheless, the classification of multipartite unitary gates is much more complicated, especially for the \emph{genuine} multipartite unitary gates, i.e. those multipartite unitaries which are not product operators across any bipartition. 
Here, we focus on the classification of genuine multipartite unitary gates of SR-2 under local equivalence by formulating their Schmidt decompositions which can be assumed in a diagonal form.


In this paper, we introduce a key notion named as \emph{singular number} (SN) to classify genuine multipartite unitary gates of SR-2 under local equivalence. The concept of SN relies on an essential observation from Lemma \ref{le:linindep} that genuine multipartite unitary gates of SR-2 have unique Schmidt decomposition. In virtue of this uniqueness, SN is defined as the number of local singular operators in the Schmidt decomposition, see Definition \ref{def:SN}. As we know, the singularity of an operator is invariant after multiplying an invertible operator. It implies that we can use the SN as an indicator of classification under local equivalence. 
The classification method is specifically described in Theorem \ref{thm:clssification}. Then we determine the accurate range of SN in Lemma \ref{le:ubofk}, which also implies the number of inequivalent classes. We specifically depict the inequivalent classes in Fig. \ref{fig:classes}. Based on this, we consider the Schmidt decompositions of genuine multiqubit unitary gates possessing SN $k$ for each $k$. We begin with a subset of genuine three-qubit unitary gates of SR-2 in Lemma \ref{le:sr2b}. Then we discuss all genuine three-qubit unitary gates of SR-2, and formulate their parametric Schmidt decompositions under local equivalence for each SN $k$ in Theorem \ref{le:sr2=singular} (i) - (iv) respectively. One can readily understand the classification given by Theorem \ref{le:sr2=singular} from Table \ref{tab:3qubit}. Furthermore, we similarly discuss the $n$-qubit scenario for $n\geq 4$, and also formulate the parametric Schmidt decompositions under local equivalence for each SN $k$ in Theorem \ref{le:multi=sr2} (i) - (v) respectively. Analogously, we also illustrate the classification given by Theorem \ref{le:multi=sr2} in Table \ref{tab:nqubit}. Comparing Theorem \ref{le:sr2=singular} and Theorem \ref{le:multi=sr2} we discover the parametric Schmidt decompositions of $n$-qubit unitary gates for $n\geq 4$ are more constrained than that of three-qubit unitary gates.  
Finally, we extend the study to three-qubit diagonal unitary gates due to the close relation between diagonal unitary gates and SR-2 unitary gates. We start with discussing two typical examples of SR-2, which helps us better understand the essential difference between the bipartite scenario and multipartite scenarios, and the core role of SN in the classification. Next, we characterize the three-qubit diagonal unitary gates of SR greater than two in Lemma \ref{le:3qubitsr3}. By further indicating that a three-qubit diagonal unitary gate has SR at most three, and presenting a necessary and sufficient condition for such a unitary gate of SR-3 in Theorem \ref{le:3qubit}, we provide a complete characterization of genuine three-qubit diagonal unitary gates.

The remainder of this paper is organized as follows. In Sec. \ref{sec:pre} we introduce the preliminaries by clarifying the notations and presenting necessary definitions and useful results. In Sec. \ref{sec:sr2class} we introduce the key notion called the singular number, and propose a method by exploiting the SN to classify genuine multipartite unitary gates of SR-2 under local equivalence. In Sec. \ref{sec:mqubitclass}, we focus on the classification of multiqubit unitary gates using the SN. For each SN $k$, we formulate the parametric Schmidt decompositions of the unitary gates possesing SN $k$ under local equivalence. In Sec. \ref{sec:3qubidiagU}, we extend the investigation to three-qubit diagonal unitary gates by characterizing the unitary gates of SR greater than two. The concluding remarks are given in Sec. \ref{sec:con}. Finally, we provide the detailed proofs of several crucial lemmas in the appendices.

\section{Preliminaries}
\label{sec:pre}

First, we clarify some notations. Denote by ${\cal H}_1\otimes...\otimes{\cal H}_n\cong\bbC^{d_1}\otimes...\otimes\bbC^{d_n}$ the $n$-partite Hilbert space whose local dimensions are $d_i$'s. A unitary matrix $U$ acting on the $n$-partite Hilbert space represents an $n$-partite unitary gate. If every local dimension $d_j=2$, $U$ is specifically called an $n$-qubit unitary gate. 
A commonly used technique to characterize multipartite unitary gates is that any $n$-partite unitary gate $U$ can be regarded as a bipartite one acting on the bipartite space ${\cal H}_S\otimes{\cal H}_{\bar{S}}$ where $S\subset\{1,2,...,n\}$ and $\bar{S}=\{1,2,...,n\}\backslash S$. In this way, the Schmidt decomposition of bipartite $U$ can be written as $U=\sum_j (B_j)_S \otimes (C_j)_{\bar{S}}$, where $(B_j)_S$ are linearly independent, and $(C_j)_{\bar{S}}$ are also linearly independent. We say that the $S$ space of $U$ is the space spanned by $(B_j)_S$, and similarly the $\bar{S}$ space of $U$ is the space spanned by $(C_j)_S$. 
Obviously, the tensor product of two unitary gates is still a unitary gate, and the properties of the latter can be characterized by that of former two gates. Therefore, to be specific, we shall focus on those unitary gates which are not the tensor product of two unitary gates with respect to any bipartition of subsystems. In this paper we call such non-product unitary gates the genuine multipartite unitary gates for simplicity.



Second, we extend the notion of \emph{Schmidt rank} (SR) to multipartite operators. We refer to the \emph{operator Schmidt rank} of an $n$-partite operator $U$ as the minimum integer $r$ such that
\begin{eqnarray}
\label{eq:u=sumr}
U=\sum^r_{j=1}
A_{j,1}\otimes...\otimes A_{j,n-1}\otimes A_{j,n},
\end{eqnarray}
where for each $j$ the operators $A_{j,i}$ act on the $i$th subsystem respectively and all have the same size. We usually call the minimum integer $r$ the \emph{Schmidt rank} of $U$ in abbreviation when there is no ambiguity.
When all matrices in Eq. \eqref{eq:u=sumr} degenerate to vectors, the definition above reduces to the Schmidt rank of multipartite pure states. When $r$ in Eq. \eqref{eq:u=sumr} reaches the minimum, we call the form on the right-hand-side of Eq.~\eqref{eq:u=sumr} as a \emph{Schmidt decomposition} of the $n$-partite operator $U$. As we shall see, the Schmidt decomposition of a multipartite operator may not be unique. Therefore, as the foundation of our classification method, we discuss a sufficient and necessary condition on the unique Schmidt decomposition for multipartite unitary gates of SR-2 in Lemma~\ref{le:linindep}. For bipartite spaces, the Schmidt decomposition in some literature additionally requires the local operators on every subsystem are orthogonal (under the Hilbert-Schmidt inner product) to each other. Specifically, in Eq. \eqref{eq:u=sumr} with $n=2$, $A_{j,1}$ for $j=1,\cdots,r$ are mutually orthogonal, and $A_{j,2}$ for $j=1,\cdots,r$ are also mutually orthogonal. However, the notion of Schmidt decomposition in this paper does not have this additional requirement.

Third, we introduce two equivalence relations between two multipartite operators. One is the local unitary (LU) equivalence. We say that two $n$-partite operators $X$ and $Y$ are LU equivalent if there are two $n$-partite product unitary matrices $V=\bigox_{i=1}^n V_i$ and $W=\bigox_{i=1}^n W_i$ such that $Y=VXW$. Such an equivalence provides convenience to prepare quantum circuits in experiments, as the unitary gate $X$ can be implemented using the prepared unitary gate $Y$ assisted with local unitary gates $V$ and $W$. Another relation is more general, which is called the equivalence under stochastic local operations and classical communications (SLOCC), or SLOCC equivalence in short. We say that two $n$-partite operators $X$ and $Y$ are SLOCC equivalent if there are two $n$-partite product invertible matrices $V=\bigox_{i=1}^n V_i$ and $W=\bigox_{i=1}^n W_i$ such that $Y=VXW$. In fact, for unitary gates the SLOCC equivalence is the same as LU equivalence, see \cite[Theorem 7]{cy14}. Thus, the two terms are used interchangeably when studing multipartite unitary gates. Such two equivalence relations also apply to quantum states. Based on the SLOCC equivalence for states, we further clarify the notion of SLOCCa equivalence for multipartite operators in Definition~\ref{def:slocca} below.
 
\begin{definition}
\label{def:slocca}
(i) {\rm \textbf{(The state corresponding to a multipartite operator)}} Suppose an $n$-partite operator $U$ acts on the $d$-dimensional system that is composed of $n$ subsystems with dimensions $d_1,d_2,\dots,d_m$, respectively, where $\prod_{j=1}^n d_j=d$. Define the state corresponding to $U$ as 
\begin{equation}
\label{eq:U-state}
(U\ox I_d)\sum_{j=1}^d \ket{j}\ket{j}_{anc}, 
\end{equation}
where the subscript ``anc'' refers to the ancilla system, and $I_d$ is the identity operator on the $d$-dimensional system. Correspondingly, the ancilla system is also composed of $n$ subsystems with dimensions $d_1,d_2,\dots,d_n$, respectively, and the $j$-th ancilla subsystem is associated with the $j$-th original subsystem.

(ii) {\rm \textbf{(SLOCCa equivalence of two multipartite operators)}} Two multipartite operators are called SLOCCa equivalent, if and only if the two states corresponding to such two operators are SLOCC equivalent. One original subsystem and its associated local ancilla are counted as one subsystem for considering the SLOCC equivalence of states here.
\end{definition}

To better understand Definition \ref{def:slocca}, we give an example of the state corresponding to the two-qubit CNOT gate. The two-qubit CNOT gate reads as $\ketbra{0,0}{0,0}+\ketbra{0,1}{0,1}+\ketbra{1,0}{1,1}+\ketbra{1,1}{1,0}$
and its corresponding state reads as $\ket{0,0,0,0}+\ket{0,0,1,1}+\ket{1,1,0,1}+\ket{1,1,1,0}$ by Eq. \eqref{eq:U-state}, where the first and third qubits are the original qubits, and the second and fourth qubits are the ancilla qubits associated with the first and third qubits respectively. Then we use the SLOCCa equivalence to discuss the SR of the CNOT gate and that of the corresponding state. One can verify that this corresponding state is of SR-2 across the bipartition of the first two qubits versus the last two qubits, and thus the CNOT gate is of operator Schmidt rank two. In general, the Schmidt rank of an operator is equal to that of its corresponding state. This directly follows from Definition \ref{def:slocca}. This equality can also be trivially extended to multipartite operators and their corresponding states. Note that the SR for multipartite states is often called the tensor rank in some literature \cite{dvc2000}. The latter usually works for pure multipartite states.



Finally, we present two important lemmas for the purpose of characterizing multipartite unitary gates of SR-2. The following lemma physically reveals that each multipartite unitary gate of SR-2 is a controlled unitary gate controlled from every subsystem, also known as a fully controlled unitary gate. 

\begin{lemma} \cite[Theorem 1]{cy13}
\label{le:cohenli2013}  
Suppose $U\in{\cal B}(\bbC^{d_1}\otimes...\otimes\bbC^{d_n})$ is an $n$-partite unitary matrix of Schmidt rank two on the $n$ subsystems $A_1,...,A_n$. Then up to the switch of subsystems, we have $U=\sum^{d_1}_{i_1=1} ... \sum^{d_{n-1}}_{i_{n-1}=1}
\ketbra{a_{1,i_1}}{b_{1,i_1}}
\otimes ... \otimes \ketbra{a_{n-1,i_{n-1}}}{b_{n-1,i_{n-1}}}
\otimes U_{i_1,...,i_{n-1}}$, where $\{\ket{a_{j,i_j}},i_j=1,...,d_j\}$ and $\{\ket{b_{j,i_j}},i_j=1,...,d_j\}$ are two orthonormal bases in $\bbC^{d_j}$ for $j=1,...,n$.
\end{lemma}

Lemma \ref{le:cohenli2013} also implies that each SR-2 unitary gate is LU equivalent to a diagonal one \cite{cy13}. Thus, we are only concerned with \emph{diagonal} unitary matrices of SR-2 in this paper. Then we review a fact in the matrix analysis, which characterizes the number of unitary matrices in the span of two diagonal unitaries.

\begin{lemma} \cite[Lemma 1]{cy16}
\label{le:spanOFdiagonal}
Suppose $U=\diag(1,d_1,...,d_{n-1})$ is a diagonal unitary matrix and is not proportional to $I_n$. Then the following two conditions are equivalent. 

(i) Up to global phases, any diagonal unitary matrix in the span of $U$ and $I_n$ must be proportional to one of $U$ and $I_n$.

(ii) There are two different numbers $d_i\neq d_j$ in the set $\{d_1,...,d_{n-1}\}$ which are both not equal to $1$.
\end{lemma}

\section{Singular number (SN) for multipartite unitary gates of SR-2}
\label{sec:sr2class}

Although the classification of two-qubit unitary gates was completed by B. Kraus \etal~\cite{kc01}, i.e., the LU equivalent expression given by Eq. \eqref{eq:twoqubit}, the classification of multipartite unitary gates is still complicated with few results. In this section we propose a method to classify multipartite unitary gates of SR-2 by the number of local singular operators in the Schmidt decomposition. This classification method is based on an essential observation that the Schmidt decomposition for genuine multipartite unitary gates of SR-2 is unique up to the switching of parties. We verify such an observation in Lemma \ref{le:linindep}. 

\begin{lemma}
\label{le:linindep}
Suppose $U$ is an $n$-partite unitary gate of Schmidt rank two.

(i) When $n=2$, the Schmidt decomposition is always not unique up to the switching of parties.

(ii) When $n\geq 3$, the Schmidt decomposition is unique up to the switching of parties if and only if $U$ up to the switching of parties cannot be decomposed as 
\begin{equation}
\label{eq:sr3unique-0}
U=(A_1\ox A_2 + B_1\ox B_2)\ox A_3\ox \cdots \ox A_n,
\end{equation}
where $A_1,~B_1$ are linearly independent, and $A_2,~B_2$ are linearly independent because $U$ is of Schmidt rank two.
\end{lemma}

\begin{proof} (i) When $n=2$, suppose $U=E_1\ox F_1+E_2\ox F_2$, where $E_1,~E_2$ are linearly independent, and $F_1,~F_2$ are linearly independent. We may always expand $U$ using linearly independent $E'_1$ and $E'_2$ which are both linear combinations of $E_1$ and $E_2$. Thus we obtain $U=E'_1\ox F'_1+E'_2\ox F'_2$. There is no other requirement for $E'_1$ and $E'_2$. So the Schmidt decomposition is always not unique, and in fact there are infinitely many forms of Schmidt decomposition.

(ii) First, we show the ``If'' part. We discuss it in two cases as follows. 

Case (ii.a) Suppose $U=A_1\otimes...\otimes A_n+B_1\otimes...\otimes B_n$, where $A_j$ and $B_j$ are linearly independent for any $j$. In order to figure out another form of Schmidt decomposition, we expand $U$ using the linear combinations of $A_j$ and $B_j$ for some $j$. Without loss of generality, we may assume $j=1$ here, and the corresponding operators on the remaining $n-1$ parties are linear combinations of $A_2\otimes...\otimes A_n$ and $B_2\otimes...\otimes B_n$. Since $A_j$ and $B_j$ are linearly independent for any $j$, it follows that any linear combination of such two corresponding operators cannot be a product one. It implies that the decomposition of $U$ into the sum of two product operators is unique up to the switching of parties. 

Case (ii.b) The remaining case is when $U$ up to the switching of parties can be written as 
\begin{equation}
\label{eq:sr3unique-1}
U=(A_1\ox A_2\ox \cdots \ox A_k + B_1\ox B_2\ox\cdots\ox B_k)\ox A_{k+1}\ox \cdots \ox A_n,
\end{equation}
where $k\geq 3$, and $A_j,~B_j$ are linearly independent for any $1\leq j\leq k$. According to the result in Case (ii.a), we obtain that $A_1\ox A_2\ox \cdots \ox A_k + B_1\ox B_2\ox\cdots\ox B_k$ in Eq. \eqref{eq:sr3unique-1} is unique, and thus the decomposition given by Eq. \eqref{eq:sr3unique-1} is unique.

Second, we show the ``Only if'' part. We prove it by contradiction. Up to the switching of parties, we may assume $U$ is written as Eq. \eqref{eq:sr3unique-0}. It follows from the assertion (i) that $A_1\ox A_2 + B_1\ox B_2$ has infinitely many forms. Thus, the Schmidt decomposition of $U$ is not unique. So we derive a contradiction, and the ``Only if'' part holds.

This completes the proof.
\end{proof}

By definition, every unitary $U$ of SR-2 can be written as $A_1\otimes...\otimes A_n+B_1\otimes...\otimes B_n$. Recall that we only consider the genuine multipartite unitary gates of SR-2. It means that $A_j$ and $B_j$ are linearly independent for any $j$. In this scenario, we conclude that the Schmidt decomposition is unique for $n\geq 3$ according to Lemma \ref{le:linindep}. Such uniqueness ensures that the definition below is well-defined.
\begin{definition}
\label{def:SN}
Suppose that $U$ is a genuine $n$-partite unitary gate of Schmidt rank two where $n\geq 3$. The singular number (SN) of $U$ is defined as the number of local singular operators in the Schmidt decomposition of $U$.
\end{definition}
Definition \ref{def:SN} also reveals that the notion of SN is a key factor to classify multipartite unitary gates of SR-2 under local equivalence, since the singularity of a matrix is invariant when multiplying it with invertible matrices. We explicitly clarify such a classification method in Theorem \ref{thm:clssification} below.

\begin{theorem}[Classification of multipartite unitary gates of Schmidt rank two]
\label{thm:clssification}
Under SLOCC equivalence, the singular number defined in Definition \ref{def:SN} is invariant when the number of parties is greater than two. For $n\geq 3$, denote by ${\cal C}_{d_1,\cdots,d_n}(k)$ the set of genuine $n$-partite unitary gates of Schmidt rank two supported on $\bbC^{d_1}\otimes...\otimes\bbC^{d_n}$ whose singular number is exactly $k$. As a result, if $U$ and $V$ respectively belong to ${\cal C}_{d_1,\cdots,d_n}(k_1)$ and ${\cal C}_{d_1,\cdots,d_n}(k_2)$ with $k_1\neq k_2$, then such two unitary gates are SLOCC inequivalent. If the local dimensions are all equal, i.e., $d_1=d_2=\cdots=d_n=d$, then we shall denote the above set as ${\cal C}(n,d,k)$ in short.
\end{theorem}

From the perspective of Theorem \ref{thm:clssification}, it is natural to ask what is the range of SN $k$. In the following we derive the accurate range of SN $k$ for the set ${\cal C}_{d_1,\cdots,d_n}(k)$. Specifically, as we shall see from Theorem \ref{le:sr2=singular} and Theorem \ref{le:multi=sr2}, for every $k$ in this range there exist unitary gates of SR-2 whose SN is exactly $k$.

\begin{lemma}
\label{le:ubofk}
For any genuine $n$-partite unitary gate of Schmidt rank two where $n\geq 3$, the singular number is at most $n$. That is $k\leq n$ for any set ${\cal C}_{d_1,\cdots,d_n}(k)$. Further, if $k\in[3,n]$, then $k\in[n-1,n]$. 
\end{lemma}

We put the proof of Lemma \ref{le:ubofk} in Appendix \ref{sec:supp}. It follows from Lemma \ref{le:ubofk} that there are four and five SLOCC inequivalent classes of genuine $n$-partite unitary gates of SR-2 for $n=3$ and $n\geq 4$ respectively according to the factor SN proposed in Theorem \ref{thm:clssification}. This result is also visualized in Fig. \ref{fig:classes}.

\begin{figure}[ht]
\centering
\includegraphics[width=5in]{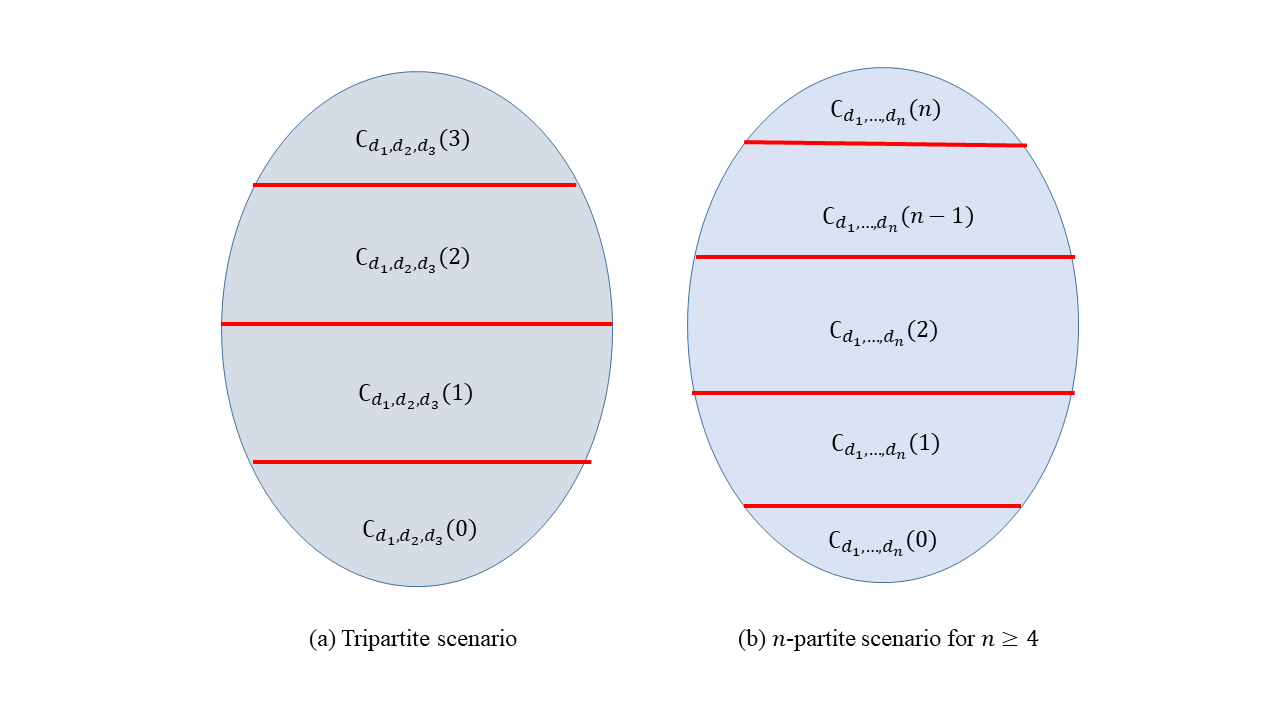}
\caption{Inequivalent classes of genuine multipartite unitary gates of SR-2 according to the singular number.}
\label{fig:classes}
\end{figure}

In the final part of this section, recall that one can regard a multipartite unitary gate as a bipartite one with respect to a bipartition of systems. In advantage of this technique, the following lemma provides a useful decomposition (not the Schmidt decomposition) to characterize genuine multipartite unitary gates of SR-2. In this way, the discussion on the $n$-partite system may be reduced to the $(n-1)$-partite system.

\begin{lemma}
\label{le:sr2}
(i) Suppose $U=P\otimes V+(I-P)\otimes W$ is an $n$-qubit unitary gate of Schmidt rank two, and $P$ is a projector on a single-qubit system. Then $V$ and $W$ both are $(n-1)$-qubit unitary gates of Schmidt rank one or two. For each case the unitary gate $U$ exists.

(ii) Suppose $U=\sum_j \proj{j}\otimes U_j$ is an $n$-partite unitary gate of Schmidt rank two, and the projector $\proj{j}$ acts on a single system. Then every $U_j$ is an $(n-1)$-partite unitary gate of Schmidt rank one or two. For each case the unitary gate $U$ exists.
\end{lemma}

\begin{proof}
(i) It follows from $U$ is unitary that $V$ and $W$ are both unitary. Because $U$ has SR-2, it directly follows that $V$ and $W$ both have SR at most two. To verify the last claim, we construct examples. Let
\begin{eqnarray}
\label{eq:multi-bi}
U=&&\proj{0}\otimes
\left[\cos\a (I_2\otimes I_2 \otimes...\otimes I_2)
+i\sin\a (\s_1\otimes\s_1\otimes...\otimes \s_1) \right]
\notag\\+&&\proj{1}\otimes
\left[\cos\b (I_2\otimes I_2 \otimes...\otimes I_2)
+i\sin\b (\s_1\otimes\s_1\otimes...\otimes \s_1) \right],
\end{eqnarray}
where $\a,\b\in[0,2\pi)$. By choosing $(\a,\b)=(0,\frac{\pi}{2}),(0,1),(1,0)$ and $(1,\frac{\pi}{4})$, one can show that $V$ and $W$ in Eq. \eqref{eq:multi-bi} has SR $(1,1)$, $(1,2),(2,1)$ and $(2,2)$ respectively, while $U$ still has SR-2.	

(ii) The assertion can be proven similarly to the proof of assertion (i).

This completes the proof.
\end{proof}

\section{Classification of genuine multiqubit unitary gates of SR-2}
\label{sec:mqubitclass}

In this section we focus on genuine multiqubit unitary gates of SR-2, and discuss the classification of such unitary gates carefully using the key notion SN clarified in Theorem \ref{thm:clssification}.
Recall that one can implement a unitary gate $U$ in experiments by applying a sequence of local unitary gates to some unitary gate $V$ which is locally equivalent to $U$. Therefore, in order to completely characterize the set ${\cal C}(n,2,k)$ for each $k$, it suffices to formulate the representative expressions of ${\cal C}(n,2,k)$, and then each unitary gate in ${\cal C}(n,2,k)$ is locally equivalent to some representative expression. 
We begin with the three-qubit system in Sec. \ref{subsec:3qubit=sr2}, and characterize the representative expressions of ${\cal C}(3,2,k)$ by formulating the parametric forms in Theorem \ref{le:sr2=singular} for every SN. Then we extend our results to $n$-qubit systems where $n\geq 4$ in Sec. \ref{subsec:nqubit=sr2}, and analogously formulate the parametric Schmidt decompositions of unitary gates in ${\cal C}(n,2,k)$ under SLOCC equivalence in Theorem \ref{le:multi=sr2} for every SN.

\subsection{Genuine three-qubit unitary gates of SR-2}
\label{subsec:3qubit=sr2}

First, we may regard the three-qubit system $ABC$ as a bipartite one with respect to one of the three bipartitions $A|BC,~B|AC$ and $C|AB$. The following lemma determines the representative expressions of those genuine three-qubit unitary gates of SR-2 satisfying the property that one of the $AB$, $AC$ and $BC$ spaces is spanned by two product unitary matrices.

\begin{lemma}
\label{le:sr2b}
Suppose $U$ is a genuine three-qubit unitary gate of Schmidt rank two acting on the system consisted of three parties $A,B,C$, and one of the $AB$, $AC$ and $BC$ spaces of $U$ is spanned by two product unitary matrices. Up to a permutation of the three parties, $U$ is LU equivalent to one of the following two forms. The first form is 
\begin{eqnarray}
\label{eq:u3qubitb-1}
U=&&\proj{0}\otimes I_2\otimes I_2+\proj{1}\otimes(\cos\a I_2+i\sin\a \s_3)\otimes (\cos\b I_2+i\sin\b \s_3),
\end{eqnarray}
and the second form is
\begin{eqnarray}
\label{eq:u3qubit-1}
U=&&\proj{0}\otimes
\left[\cos\a (I_2\otimes I_2)
+i\sin\a (\s_3\otimes\s_3)\right]
\notag\\+&&\proj{1}\otimes
\left[\cos\b (I_2\otimes I_2)
+i\sin\b (\s_3\otimes\s_3)\right],
\end{eqnarray}
where $\a$ and $\b$ are both in $[0,2\pi)$ such that $U$ is of Schmidt rank two.
\end{lemma}

\begin{proof}
Up to a permutation of the three parties, we may assume that the $BC$ space of $U$ is spanned by two product unitary matrices. It follows from Lemma \ref{le:cohenli2013} that any unitary matrix of SR-2 is LU equivalent to a diagonal one. Then under LU equivalence we may assume that the two product unitary matrices spanning the $BC$ space are $I_2\otimes I_2$ and $\diag(1,x)\otimes\diag(1,y)$ with $x,y$ of modulus one. Thus, the three-qubit unitary gate $U$ can be decomposed as
\begin{eqnarray}
U=A_1\otimes I_2\otimes I_2
+A_2\otimes \diag(1,x)\otimes\diag(1,y).	
\end{eqnarray}
Because $U$ is controlled from the first system from Lemma \ref{le:cohenli2013}, we may also assume that $A_1$ and $A_2$ are both diagonal matrices. Let $A_1=\diag(a,b)$ and $A_2=\diag(c,d)$. Note that $a,b,c,d$ have at most two zeros. Applying the unitary equivalence to system $A$, we obtain four cases: (i) $b=c=0$, (ii) $c=d=0$, (iii) $d=0$ and (iv) $abcd\ne0$. In case (i), $U$ is LU equivalent to Eq. \eqref{eq:u3qubitb-1}. In case (ii), $U$ is a product unitary gate, so this case does not exist. In case (iii), since $U$ is controlled from system $A$, there must be a unitary which is a linear combination of $I_2\otimes I_2$ and $\diag(1,x)\otimes\diag(1,y)$ with two nonzero coefficients. Then Lemma~\ref{le:spanOFdiagonal} implies that $x=y=-1$. So $U$ is LU equivalent to Eq. \eqref{eq:u3qubit-1} with $\b=0$. In case (iv), the similar argument also applies, and we obtain that $U$ is LU equivalent to Eq. \eqref{eq:u3qubit-1}. On the other hand, evidently both Schmidt decompositions given by Eqs. \eqref{eq:u3qubitb-1} - \eqref{eq:u3qubit-1} need satisfy the hypothesis of this lemma that $U$ is of SR-2. This completes the proof.
\end{proof}

Furthermore, we consider the classification of all genuine three-qubit unitary gates of SR-2. Specifically, in Theorem \ref{le:sr2=singular} we derive the parametric forms of the Schmidt decomposition under SLOCC (LU) equivalence for the unitary gates in the set ${\cal C}(3,2,k)$ for each SN $k$. 
To characterize the set of unitary gates with SN $k=1$, i.e., ${\cal C}(3,2,1)$, we have to present the following lemma first.

\begin{lemma}
\label{le:3qbsr2c3}
Given a positive number $c\neq 1$, there exist three nonzero complex numbers $f,~g,~h$ with $f\neq g$ and $h\neq 1$ such that the following continuous equality
\begin{equation}
\label{eq:3qbsr2c3-1}
\abs{f+c}=\abs{g+c}=\abs{fh+c}=\abs{gh+c}=1
\end{equation}
holds, if and only if $f,~g,~h$ are given in one of the following two cases:

(i) $f=e^{i\a}-c$, $g=f^*$, and $h=\frac{c^2-1}{1+c^2-2c\cos\a}$. Here, the free parameter $\a$ satisfies that $\a\in(0,\pi)\cup(\pi,2\pi)$ and $c\cos\a\neq 1$.

(ii) $f=e^{i\a}-c$, $g=-e^{i(\t-\a)}-c$, where $\t=2\arctan\big(\frac{c\sin\frac{\a+\gamma}{2}-\sin\frac{\a-\gamma}{2}}{c\cos\frac{\a+\gamma}{2}-\cos\frac{\a-\gamma}{2}}\big)\in(-\pi,\pi)$, and $h=\frac{e^{i\gamma}-c}{e^{i\a}-c}$.
Here, due to the periodicity we may assume that $\a,\gamma\in[0,2\pi)$, and such two free parameters satisfy that $\a\notin\{\gamma,\frac{\pi+\t}{2},\frac{3\pi+\t}{2}\}$, and $c\cos\frac{\a+\gamma}{2}-\cos\frac{\a-\gamma}{2}\neq 0$.
\end{lemma}

We put the proof of Lemma \ref{le:3qbsr2c3} in Appendix \ref{sec:3qbsr2c3}.

Now, we are able to describe the classification of genuine three-qubit unitary gates of SR-2 by the key factor SN in Theorem \ref{le:sr2=singular} below. In order to understand the main result Theorem \ref{le:sr2=singular} conveniently, we also illustrate the classification in Table \ref{tab:3qubit}.

\begin{table}[!h]
\caption{The classification of genuine three-qubit unitary gates of SR-2 under local equivalence}
\centering
\begin{tabular}{|c|p{10.5cm}|p{4cm}|}
\hline
singular number & parametric Schmidt decomposition & range of parameters  \\
\hline
$k=3$ & $I_2\otimes I_2\otimes I_2+(e^{i\ph}-1)\proj{0,0,0}$ & $\ph\in(0,2\pi)$  \\
\hline
\multirow{2}*{$k=2$} & $I_2\otimes I_2\otimes I_2
+\proj{1}\otimes \proj{1}\otimes \diag(e^{i\theta}-1,e^{i\phi}-1)$ & $\theta,\phi\in(0,2\pi)$  \\
\cline{2-3}
  & $\proj{0}\otimes I_2\otimes I_2+\proj{1}\otimes\diag(1,e^{i\g})\otimes\diag(1,e^{i\d})$ & $\g,\d\in[0,2\pi)$ \\
\hline
\multirow{2}*{$k=1$} & $\proj{0}\otimes \diag(e^{i\a}-c,e^{-i\a}-c)\otimes \diag(1,\frac{c^2-1}{1+c^2-2c\cos\a})+\diag(c,1)\ox I_2\ox I_2$ & $\a\in(0,\pi)\cup(\pi,2\pi)$, $c\in(0,1)\cup(1,+\infty)$, $c\cos\a\neq 1$  \\
\cline{2-3}
  & $\proj{0}\otimes \diag(e^{i\a}-c,-e^{i(\t-\a)}-c)\otimes \diag(1,\frac{e^{i\gamma}-c}{e^{i\a}-c})+\diag(c,1)\ox I_2\ox I_2$, $\t=2\arctan\big(\frac{c\sin\frac{\a+\gamma}{2}-\sin\frac{\a-\gamma}{2}}{c\cos\frac{\a+\gamma}{2}-\cos\frac{\a-\gamma}{2}}\big)\in(-\pi,\pi)$ & $c\in(0,1)\cup(1,+\infty)$, $\a,\g\in[0,2\pi)$, $\a\notin\{\gamma,\frac{\pi+\t}{2},\frac{3\pi+\t}{2}\}$, and $c\cos\frac{\a+\gamma}{2}\neq \cos\frac{\a-\gamma}{2}$  \\
\hline
$k=0$ & $\diag(a,b)\ox\diag(1,c)\ox\diag(1,d)+\diag(1-a,1-b)\ox \diag(1,\frac{1-bc}{1-b})\ox\diag(1,\frac{1-bd}{1-b})$ & $a,b,c,d\in\bbC\backslash\{0,1\}$ with $a\neq b$ are a solution of Eq. \eqref{eq:u3qubiti} \\
\hline
\end{tabular}
\label{tab:3qubit}
\end{table}

\begin{theorem}
\label{le:sr2=singular}
Suppose $U\in {\cal C}(3,2,k)$ is a genuine three-qubit unitary gate of Schmidt rank two with singular number $k$. Then $0\le k\le 3$. For each $k$, the representative expressions of ${\cal C}(3,2,k)$, i.e., the Schmidt decompositions under local unitary equivalence are parameterized as follows.

(i) For $k=3$, up to a permutation of systems and under local equivalence the Schmidt decomposition of $U$ is $I_2\otimes I_2\otimes I_2+(e^{i\ph}-1)\proj{0,0,0}$, where $\ph\in(0,2\pi)$.

(ii) For $k=2$, up to a permutation of systems and under local equivalence the Schmidt decomposition of $U$ is either
\begin{eqnarray}
\label{eq:n=2,U1}
I_2\otimes I_2\otimes I_2
+\proj{1}\otimes \proj{1}\otimes \diag(e^{i\theta}-1,e^{i\phi}-1),\quad \theta,\phi\in(0,2\pi),
\end{eqnarray}
or
\begin{eqnarray}
\label{eq:n=2,U2}
\proj{0}\otimes I_2\otimes I_2+\proj{1}\otimes\diag(1,e^{i\g})\otimes\diag(1,e^{i\d}), \quad \g,\d\in[0,2\pi).
\end{eqnarray}

(iii) For $k=1$, up to a permutation of systems and under local equivalence the Schmidt decomposition of $U$ is either
\begin{equation}
\label{eq:n=1,claim}	
\proj{0}\otimes \diag(e^{i\a}-c,e^{-i\a}-c)\otimes \diag(1,\frac{c^2-1}{1+c^2-2c\cos\a})+\diag(c,1)\ox I_2\ox I_2,
\end{equation}
where $\a\in(0,\pi)\cup(\pi,2\pi)$ and $c\cos\a\neq 1$ for some positive $c\neq 1$, or
\begin{equation}
\label{eq:n=1,claim-2}  
\proj{0}\otimes \diag(e^{i\a}-c,-e^{i(\t-\a)}-c)\otimes \diag(1,\frac{e^{i\gamma}-c}{e^{i\a}-c})+\diag(c,1)\ox I_2\ox I_2,
\end{equation}
where $\t=2\arctan\big(\frac{c\sin(\frac{\a+\gamma}{2})-\sin(\frac{\a-\gamma}{2})}{c\cos(\frac{\a+\gamma}{2})-\cos(\frac{\a-\gamma}{2})}\big)$, and for some given positive $c\neq 1$ the parameters $\a,\gamma$ satisfy that $\a\neq \gamma$, $\a\neq\frac{\pi+\t}{2}$, and $c\cos(\frac{\a+\gamma}{2})-\cos(\frac{\a-\gamma}{2})\neq 0$.

(iv) For $k=0$, up to a permutation of systems and under local equivalence the Schmidt decomposition of $U$ is 
\begin{equation}
\label{eq:sr23qubit}
\diag(a,b)\ox\diag(1,c)\ox\diag(1,d)+\diag(1-a,1-b)\ox \diag(1,\frac{1-bc}{1-b})\ox\diag(1,\frac{1-bd}{1-b}),
\end{equation}
where the four parameters $a,b,c,d\in\bbC\backslash\{0,1\}$ with $a\neq b$ constitute a solution of the following system of equations 
\begin{equation}
\label{eq:u3qubiti}
\left\{
\begin{aligned}
\abs{(1-a)(1-d)+d(1-b)}&=\abs{1-b}, \\
\abs{(1-a)(1-c)+c(1-b)}&=\abs{1-b}, \\
\abs{(1-a)(1-bc)(1-bd)+acd(1-b)^2}&=\abs{1-b}^2, \\
\abs{(1-bc)(1-bd)+bcd(1-b)}&=\abs{1-b}.
\end{aligned}
\right.
\end{equation}
\end{theorem}

\begin{proof}
It directly follows from Lemma \ref{le:ubofk} that $k\leq 3$. 
Suppose $U=A_1\otimes A_2\otimes A_3+B_1\otimes B_2\otimes B_3$ is the Schmidt decomposition of $U$. Since $U$ is of SR-2, it follows from Lemma \ref{le:cohenli2013} that we may assume that $A_j$ and $B_j$ are all diagonal matrices under local equivalence. Next, we discuss the representative expressions of ${\cal C}(n,2,k)$ by the key factor SN $k$.


(i) For $k=3$, we may assume $B_1,~B_2$ are singular without loss of generality. It follows that $A_1,~A_2,~A_3$ all have to be unitary. Thus, we can further assume that $B_1,B_2,B_3$ are all singular. By applying proper local unitary gates, we may assume that $A_j=\diag(1,1)$ for every $j$ and $B_1\otimes B_2\otimes B_3=x \proj{0,0,0}$ with $\abs{x+1}=1$. It follows that $x+1=e^{i\ph}$ and the assertion holds.

(ii) For $k=2$, up to a permutation of systems there are only two different cases, i.e., either $A_1,~A_2$ are both singular or $A_1,~B_1$ are both singular. According to the assumption that $A_j$ and $B_j$ are all diagonal, we assume that either $A_1=A_2=\proj{0}$ or $A_1=\proj{0}$ and $B_1=\proj{1}$ under local equivalence.

In the case when $A_1=A_2=\proj{0}$, we obtain
\begin{eqnarray}
\label{eq:n=2,U1a}	
U=\proj{0}\otimes \proj{0}\otimes \diag(a_0,b_0)
+\diag(1,e^{i\a})\otimes\diag(1,e^{i\b})\otimes\diag(1,e^{i\gamma}),
\end{eqnarray}
where $\a,\b,\gamma\in [0,2\pi)$, and $a_0,b_0\in\mathbb{C}$. The last two matrices in the second term are unitary because $A_1$ is singular, and the first and third matrices in the second term are unitary because $A_2$ is singular. We may further apply diagonal phase gates on all three qubits, such that the global gate becomes
\begin{eqnarray}
\label{eq:n=2,U1b}	
U=\proj{0}\otimes \proj{0}\otimes \diag(a,b) +I_2\ox I_2\ox I_2,
\end{eqnarray}
where $a,b\in\mathbb{C}$ and $ab\ne 0$ since $A_3$ is not singular. Then, after swapping $\ket{0}$ and $\ket{1}$ on each of the first two qubits, and making use of the facts that $U$ is unitary, and that $ab\ne 0$, we obtain the form in Eq. \eqref{eq:n=2,U1}.

In the case when $A_1=\proj{0}$ and $B_1=\proj{1}$, it follows that $A_2,~A_3,~B_2,~B_3$ are all diagonal unitary. Then, after applying proper local unitary gates on last two qubits, we obtain the form in Eq. \eqref{eq:n=2,U2}. 

(iii) For $k=1$, up to a permutation of systems and under local equivalence we may assume that $A_1=\proj{0}$, $B_1=c\proj{0}+\proj{1}$ and $B_2=B_3=I_2$. It follows that
\begin{eqnarray}
\label{eq:n=1}	
U=\proj{0}\otimes \diag(f,g)\otimes \diag(1,h) +(c\proj{0}+\proj{1})\ox I_2\ox I_2,
\end{eqnarray}
where $f,g,h\in\mathbb{C}\backslash\{0\}$, $c>0$, $f\ne g$ and $h\ne 1$ because $U$ is not a product gate across any bipartition of the three-qubit system. In order to make $c>0$, there could be a phase for $\proj{0}$ in the first term, but such phase may be absorbed into $f$ and $g$. For $U$ is unitary, it requires that the $4\times 4$ matrix
\begin{equation}
\label{eq:n=1,b}
V=\diag(f,g)\ox \diag(1,h)+ c I_2\ox I_2=\diag(c+f,c+fh,c+g,c+gh)
\end{equation}
is unitary. It is equivalent to that the following continuous equality holds
\begin{equation}
\label{eq:n=1,conequ-1}
\abs{c+f}=\abs{c+fh}=\abs{c+g}=\abs{c+gh}=1.
\end{equation}
We may further assume that
\begin{eqnarray}
\label{eq:angledefs}	
c+f=e^{i\a}, \quad c+g=e^{i\b},\quad c+fh=e^{i\gamma}, \quad c+gh=e^{i\delta},
\end{eqnarray}
where $\a,\b,\gamma,\delta\in [0,2\pi)$. When $c=1$, it follows from $h=\frac{e^{i\gamma}-1}{e^{i\a}-1}=\frac{e^{i\delta}-1}{e^{i\beta}-1}$ that 
\begin{equation}
\label{eq:angles2-01}
e^{i(\b+\gamma)}-(e^{i\b}+e^{i\gamma})=e^{i(\a+\delta)}-(e^{i\a}+e^{i\delta}).
\end{equation}
From Lemma \ref{le:n=3k=1xneq1} in Appendix \ref{sec:supp} we conclude that Eq. \eqref{eq:angles2-01} holds only when $\a=\b$ or $\a=\gamma$. However, the solutions of Eq. \eqref{eq:angles2-01} contradict to the requirements $f\neq g$ and $h\neq 1$. It means there are no satisfied $f,g,h$ such that the continuous equality \eqref{eq:n=1,conequ-1} holds if $c=1$. For any given positive $c\neq 1$, all satisfied $f,g,h$ have been formulated in Lemma \ref{le:3qbsr2c3}. The two parametric Schmidt decompositions given by Eqs. \eqref{eq:n=1,claim} and \eqref{eq:n=1,claim-2} correspond to the two cases in Lemma \ref{le:3qbsr2c3} respectively.

(iv) For $k=0$, we may similarly assume the Schmidt decomposition of $U$ as 
\begin{eqnarray}
\label{eq:u3qubitf}
U=\diag(a,b)\ox\diag(1,c)\ox\diag(1,d)+\diag(v,f)\ox \diag(1,g)\ox\diag(1,h),
\end{eqnarray}
where $a,b,c,d,v,f,g,h\in\mathbb{C}\backslash\{0\}$.
Moreover, $U$ can also be decomposed as $U=\proj{0}\ox V+ \proj{1}\ox W$, where $V$ and $W$ are both diagonal unitary operators of SR-2 because otherwise $k\ge 1$. Then under LU equivalence we may assume
\begin{eqnarray}
\label{eq:u3qubite}
U=\diag(1,e^{i\a},e^{i\b},e^{i\gamma},1,1,1,e^{i\delta}),
\end{eqnarray}
where $\a,\b,\gamma,\delta\in[0,2\pi)$, and $\delta\ne 0$ because $W$ is of SR-2.

Note that Eqs. \eqref{eq:u3qubitf} and \eqref{eq:u3qubite} are compatible. This means that for any diagonal three-qubit unitary gate $U_0$, there are $2\times 2$ unitary matrices $V_1,~V_2,~V_3$ such that $U=(V_1\ox V_2 \ox V_3)U_0$ satisfies both Eqs. \eqref{eq:u3qubitf} and \eqref{eq:u3qubite}. This is because the form in Eq.~\eqref{eq:u3qubite} may be obtained from a general diagonal three-qubit unitary gate by multiplying it with a local unitary gate $\diag(1,e^{i\theta_1})\ox \diag(1,e^{i\theta_2})\ox \diag(1,e^{i\theta_3})$, where $\theta_j\in [0,2\pi)$, $j=1,2,3$. Such a local diagonal unitary operation does not change the general form in Eq.~\eqref{eq:u3qubitf} but only changes the values of $a,b,c,d,v,f,g,h$. 

By comparing Eq.~\eqref{eq:u3qubitf} with Eq.~\eqref{eq:u3qubite}, we obtain $a\diag(1,d,c,cd)+v\diag(1,h,g,gh)=\diag(1,e^{i\a},e^{i\b},e^{i\gamma})$, and $b\diag(1,d,c,cd)+f\diag(1,h,g,gh)=\diag(1,1,1,e^{i\delta})$, where $\delta\ne 0$. It follows that $v=1-a$, $f=1-b$, $g=\frac{1-bc}{1-b}$, $h=\frac{1-bd}{1-b}$, and
\begin{equation}
\label{eq:u3qubith}
\left\{
\begin{aligned}
ad+\frac{(1-a)(1-bd)}{1-b}&=e^{i\a}, \\
ac+\frac{(1-a)(1-bc)}{1-b}&=e^{i\b}, \\
acd+\frac{(1-a)(1-bc)(1-bd)}{(1-b)^2}&=e^{i\gamma}, \\
bcd+\frac{(1-bc)(1-bd)}{1-b}&=e^{i\delta},
\end{aligned}
\right.
\end{equation}
where $a\ne 1$ and $b\ne 1$ because $v$ and $f$ are nonzero, and $abcd\delta\ne 0$ should still hold. Moreover, since $U$ is genuine, from Eq. \eqref{eq:u3qubitf} we conclude that $a,b,c,d$ additionally satisfy $\frac{a}{b}\neq \frac{v}{f}$, $c\neq g$ and $d\neq h$. By direct calculation, these constraints are indeed $a\neq b$, $c\neq 1$ and $d\neq 1$. One can verify that $\delta$ cannot be zero under these constraints for $a,b,c,d$. 
By noting that $\abs{e^{i\a}}=\abs{e^{i\b}}=\abs{e^{i\gamma}}=\abs{e^{i\delta}}=1$, we obtain the system of equations as Eq. \eqref{eq:u3qubiti}.

This completes the proof.
\end{proof}


Here, we would like to give some remarks on the clssification of genuine three-qubit unitary gates of SR-2 by comparing the parametric Schmidt decompositions formulated in Theorem \ref{le:sr2=singular} with the two forms given by Eq.~\eqref{eq:u3qubitb-1} and Eq.~\eqref{eq:u3qubit-1} in Lemma \ref{le:sr2b}.
First, one can verify the expression of $U$ in Theorem \ref{le:sr2=singular} (i) is not locally equivalent to the form in Eq.~\eqref{eq:u3qubitb-1} or Eq.~\eqref{eq:u3qubit-1} by direct calculation. Second, on the one hand, the expression of $U$ given by Eq.~\eqref{eq:n=2,U1} in Theorem \ref{le:sr2=singular} (ii) is also not locally equivalent to the form in Eq.~\eqref{eq:u3qubitb-1} or Eq.~\eqref{eq:u3qubit-1}. This can be deduced from Lemma~\ref{le:spanOFdiagonal}, as there are three different diagonal entries in the expression Eq.~\eqref{eq:n=2,U1}. It implies that none of the $AB,~AC$, and $BC$ spaces of $U$ given by Eq.~\eqref{eq:n=2,U1} is spanned by two product unitary matrices. On the other hand, the expression of $U$ given by Eq. \eqref{eq:n=2,U2} is locally equivalent to the form in Eq.~\eqref{eq:u3qubitb-1} by observation. Third, some but not all expressions of $U$ in Theorem \ref{le:sr2=singular} (iii) are locally equivalent to the form in Eq.~\eqref{eq:u3qubit-1}. For example, $U=i\sqrt{1-c^2}\proj{0}\otimes\sigma_3\otimes\sigma_3 +(c\proj{0}+\proj{1})\ox I_2\ox I_2$ for $c\in(0,1)$ is a special form of Eq. \eqref{eq:n=1,claim}, and it is also a special form of Eq.~\eqref{eq:u3qubit-1} under local equivalence. However, for $c>1$, the unitary gate $U$ given by Eq. \eqref{eq:n=1,claim} is not locally equivalent to the form in Eq.~\eqref{eq:u3qubit-1}. Finally, the parametric Schmidt decomposition formulated in Theorem \ref{le:sr2=singular} (iv) is locally equivalent to the form in Eq.~\eqref{eq:u3qubit-1} if and only if the $BC$ space of $U$ is spanned by two product unitary matrices $I_2\ox I_2$ and $\sigma_3\ox\sigma_3$.

\subsection{Genuine multiqubit unitary gate of SR-2}
\label{subsec:nqubit=sr2}

Next, we extend the classification of genuine three-qubit unitary gates of SR-2 to multiqubit scenarios. We similarly utilize the key notion SN to present a complete classification of genuine multiqubit unitary gates of SR-2 by explicitly expressing the parametric Schmidt decompositions of unitary gates in ${\cal C}(n,2,k)$ for $n\geq 4$. The classification is specifically described in Theorem \ref{le:multi=sr2} below.
Note that in Theorem \ref{le:sr2=singular} (iv) for the case of SN $k=0$, the Schmidt decomposition is parameterized in an implicit form, where the parameters are contained in a system of equations. Nevertheless, as we shall see from Theorem \ref{le:multi=sr2}, the Schmidt decompositions of unitary gates in ${\cal C}(n,2,k)$ are parameterized in explicit expressions for every SN $k$ when $n\geq 4$. Thus, the three-qubit case seems to be the most difficult case in the classification of genuine multiqubit unitary gates of SR-2. Before presenting our main result Theorem \ref{le:multi=sr2}, we need the following lemma to analyze the case of SN $k=0$.

\begin{lemma}
\label{le:ngeq4,k=0}
Suppose $U=\proj{0}\ox G+\proj{1}\ox H$ is a genuine $n$-qubit unitary gate of Schmidt rank two with singular number $k=0$, where $n\geq 4$.

(i) Both $G$ and $H$ are $(n-1)$-qubit unitary gates of Schmidt rank two, and are not product unitary gates across any bipartition of the last $(n-1)$ qubits.

(ii) Suppose $G$ and $H$ are both diagonal under local equivalence, and $G$ has the Schmidt decomposition $G=\diag(a,b)\ox C_3\ox\cdots\ox C_n+\diag(c,d)\ox D_3\ox\cdots\ox D_n$, where $C_3\ox\cdots\ox C_n$ and $D_3\ox\cdots\ox D_n$ are fixed for $G$, and are linearly independent from each other. Then $H$ can also be expanded with $C_3\ox\cdots\ox C_n$ and $D_3\ox\cdots\ox D_n$, i.e., $H$ has the Schmidt decomposition $H=\diag(p,q)\ox C_3\ox\cdots\ox C_n+\diag(r,s)\ox D_3\ox\cdots \ox D_n$.
\end{lemma}

We put the proof of Lemma \ref{le:ngeq4,k=0} in Appendix \ref{sec:ngeq4,k=0}.

Then we present the complete classification of genuine $n$-qubit unitary gates of SR-2 in Theorem \ref{le:multi=sr2} below. Analogously, in order to readily understand this main result we also illustrate the classification in Table \ref{tab:nqubit}.

\begin{table}[!h]
\caption{The classification of genuine $n$-qubit ($n\geq 4$) unitary gates of SR-2 under local equivalence}
\centering
\begin{tabular}{|c|c|c|}
\hline
singular number & parametric Schmidt decomposition & range of parameters  \\
\hline
$k=n$ & $I_2\otimes\cdots\otimes I_2+(e^{i\theta}-1)\proj{0}^{\otimes n}$ & $\t\in(0,2\p)$  \\
\hline
$k=n-1$ & $I_2\otimes\cdots\otimes I_2+\proj{0}^{\otimes n-1}\otimes \diag(e^{i\theta}-1,e^{i\phi}-1)$ & $\theta,\phi\in(0,2\pi)$ and $\theta\ne\phi$  \\
\hline
$k=2$ & $U=\proj{0}\otimes I_2\otimes \cdots \otimes I_2+\proj{1}\otimes\diag(1,e^{i\b_2})\otimes\cdots\otimes\diag(1,e^{i\b_n})$ & $\b_2,\cdots,\b_n\in(0,2\pi)$  \\
\hline
$k=1$ & $U=i\cos\a\proj{0}\ox \sigma_3^{\ox (n-1)} + \diag(\sin\a,1)\ox I_2^{\ox (n-1)}$ & $\a\in(0,\frac{\pi}{2})\cup(\frac{\pi}{2},\pi)$  \\
\hline
$k=0$ & $U=\diag(\cos\a,\cos\b)\ox I_2^{\ox (n-1)}+i\diag(\sin\a,\sin\b)\ox\sigma_3^{\ox (n-1)}$ & $\a,\b\neq \frac{k\pi}{2}$ for $k=0,1,2,3$  \\
\hline
\end{tabular}
\label{tab:nqubit}
\end{table}

\begin{theorem}
\label{le:multi=sr2}
Suppose $U\in {\cal C}(n,2,k)$ is a genuine $n$-qubit unitary gate of Schmidt rank two with singular number $k$, where $n\geq 4$ and $k\in[0,2]\cup [n-1,n]$. For each $k$, the representative expressions of ${\cal C}(n,2,k)$, i.e., the Schmidt decompositions under local unitary equivalence are parameterized as follows.

(i) For $k=n$, up to a permutation of systems and under local equivalence the Schmidt decomposition of $U$ is 
\begin{eqnarray}
\label{eq:multi,k=n}
I_2\otimes\cdots\otimes I_2+(e^{i\theta}-1)\proj{0}^{\otimes n},
\end{eqnarray}
where $\t\in(0,2\p)$.

(ii) For $k=n-1$, up to a permutation of systems and under local equivalence the Schmidt decomposition of $U$ is
\begin{eqnarray}
\label{eq:multi,k=n-1}
I_2\otimes\cdots\otimes I_2+\proj{0}^{\otimes n-1}\otimes \diag(e^{i\theta}-1,e^{i\phi}-1),
\end{eqnarray}
where $\theta,\phi\in(0,2\pi)$ and $\theta\ne\phi$.

(iii) For $k=2$, up to a permutation of systems and under local equivalence the Schmidt decomposition of $U$ is
\begin{eqnarray}
\label{eq:multi,k=2,U2}
U=\proj{0}\otimes I_2\otimes \cdots \otimes I_2+\proj{1}\otimes\diag(1,e^{i\b_2})\otimes
\cdots\otimes\diag(1,e^{i\b_n}),
\end{eqnarray}
where $\b_2,\cdots,\b_n\in(0,2\pi)$.

(iv) For $k=1$, up to a permutation of systems and under local equivalence the Schmidt decomposition of $U$ is
\begin{eqnarray}
\label{eq:multi,k=1,U0}
U=i\cos\a\proj{0}\ox \sigma_3^{\ox (n-1)} + \diag(\sin\a,1)\ox I_2^{\ox (n-1)},
\end{eqnarray}
where $\a\in(0,\frac{\pi}{2})\cup(\frac{\pi}{2},\pi)$. 

(v) For $k=0$, up to a permutation of systems and under local equivalence the Schmidt decomposition of $U$ is
\begin{eqnarray}
\label{eq:multi,k=0,U0}
U=\diag(\cos\a,\cos\b)\ox I_2^{\ox (n-1)}+i\diag(\sin\a,\sin\b)\ox\sigma_3^{\ox (n-1)}
\end{eqnarray}
where $\a,\b\neq \frac{k\pi}{2},~k=0,1,2,3$.
\end{theorem}

\begin{proof}
Suppose $U=A_1\otimes ... \otimes A_n+B_1\otimes ... \otimes B_n$ is the Schmidt decomposition, where $A_i,~B_j,~\forall 1\leq i,j\leq n$ are all diagonal matrices under local equivalence. Theorem \ref{le:sr2=singular} analyze the case of $n=3$. Here, we focus on the case of $n\geq 4$. From Lemma \ref{le:ubofk} we determine that the SN $k$ is upper bounded by the number of parties $n$, and further $k$ is either $n$ or $n-1$ if $k\in[3,n]$.

(i) For $k=n$, we conclude that either $A_1,\cdots,A_n$ are all singular or $B_1,\cdots,B_n$ are all singular. Otherwise, suppose there is an $A_i$ and a $B_j$ are both singular. It follows that all $A_s$ with $s\neq j$ are all unitary, and all $B_t$ with $t\neq i$ are all unitary. It leads to a contradiction that the SN $k=2$. Thus, without loss of generality, we may assume $B_1,\cdots,B_n$ are all singular, and thus $U=I_2^{\otimes n}+c\proj{0}^{\otimes n}$ under local equivalence. It implies that $c=e^{i\theta}-1$ where $\theta\in(0,2\pi)$. This gives the form in Eq. \eqref{eq:multi,k=n}.

(ii) For $k=n-1$, from the proof of Lemma \ref{le:ubofk} we may assume $U=I_2^{\otimes n}+c\proj{0}^{\otimes (n-1)}\ox \diag(1,t_1)$ under local equivalence and up to a permutation of systems. It follows that $c=e^{i\theta}-1$, $ct_1=e^{i\phi}-1$, where $\theta,\phi\in(0,2\pi)$, and $\theta\ne\phi$ because $U$ is a genuine $n$-qubit unitary gate. This gives the form in Eq. \eqref{eq:multi,k=n-1}.

(iii) For $k=2$, up to a permutation of systems there are two different cases for the Schmidt decomposition of $U$. The first case is when $B_1,~B_2$ are both singular, and the second case is when $A_1,~B_1$ are both singular. In the first case, we may similarly assume $U$ as Eq. \eqref{eq:multi,k=2,U}. It follows from the proof of Lemma \ref{le:ubofk} that only $k\geq n-1$ is possible. It means $n\leq k+1=3$. Thus, this case is reduced to Eq. \eqref{eq:n=2,U1} in Theorem (ii).  
For the second case, we determine that $U$ is locally equivalent to Eq. \eqref{eq:multi,k=2,U2}, because $A_i,~B_i,~\forall 2\leq i\leq n$, are all $2\times 2$ diagonal unitary matrices, and all of them are not $I_2$ for $U$ is a genuine $n$-qubit unitary gate.

(iv) For $k=1$, up to a permutation of systems and under local equivalence we may assume 
\begin{eqnarray}
\label{eq:multi,k=1,U}
U=\proj{0}\ox A_2\ox\cdots\ox A_n + (c\proj{0}+\proj{1})\ox B_2\ox\cdots\ox B_n,
\end{eqnarray}
where $c\in\mathbb{R}\backslash\{0\}$, and $A_i,~B_i,~\forall 2\le i\le n$, are all non-singular. It follows that $B_2,\dots,B_n$ are all diagonal unitary. Hence, under local equivalence we further assume
\begin{eqnarray}
\label{eq:multi,k=1,U2}
U=\proj{0}\ox A_2\ox\cdots\ox A_n + (c\proj{0}+\proj{1})\ox I_2^{\ox (n-1)}.
\end{eqnarray}
It follows that
\begin{eqnarray}
\label{eq:multi,k=1,U2b}
V:=A_2\ox\cdots\ox A_n + c I_2^{\ox (n-1)}
\end{eqnarray}
is an $(n-1)$-qubit unitary gate.
It implies that
\begin{eqnarray}
\label{eq:multi,k=1,U3}
A_2^\dag A_2\ox\cdots\ox A_n^\dag A_n + c A_2\ox\cdots\ox A_n + c A_2^\dag\ox\cdots\ox A_n^\dag=(1-c^2) I_2^{\ox (n-1)}.
\end{eqnarray}
Suppose a diagonal entry of $A_2\ox\cdots\ox A_n$ is $x$. Then Eq.~\eqref{eq:multi,k=1,U3} implies $\vert x\vert^2+2c\textrm{Re}(x)=1-c^2$, or equivalently,
\begin{eqnarray}
\label{eq:multi,k=1,xc}
\vert x+c\vert=1.
\end{eqnarray}
The case of $n=3$ is discussed in Theorem~\ref{le:sr2=singular}(iii).
Here, we further discuss the case of $n\ge 4$. Obviously, Eq. \eqref{eq:multi,k=1,xc} has a class of solutions as $x=i\cos\a$ and $c=\sin\a$. This class of solutions corresponds to the Schmidt decomposition as Eq. \eqref{eq:multi,k=1,U0}, where $\a\in(0,\frac{\pi}{2})\cup(\frac{\pi}{2},\pi)$. We aim to show the claim that there is no other type of Schmidt decomposition. We may extract a constant factor from each $A_i$ such that $A_2\ox\cdots\ox A_n=g A'_2\ox\cdots\ox A'_n$, where $g\in\mathbb{C}\backslash\{0\}$, and $A'_i=\diag(1,a_i)$ with $a_i\ne 0,1$, for $2\le i\le n$. By observation the Schmidt decomposition of $U$ given by Eq. \eqref{eq:multi,k=1,U0} corresponds to $a_i=-1$, for $2\le i\le n$. We prove the above claim by contradiction, and assume the diagonal entries $a_i$ in each $A'_i$ are not identically $-1$. 

Without loss of generality, we may assume $a_3\ne -1$. Consider the eight diagonal entries of $g A'_2\ox\cdots\ox A'_n$ acting on the last $n-4$ qubits. Explicitly, they are
\begin{eqnarray}
\label{eq:multi,k=1,elements}
g(1,a_4,a_3,a_3 a_4, a_2, a_2 a_4, a_2 a_3, a_2 a_3 a_4).
\end{eqnarray}
Similar to the discussion in the proof of Lemma \ref{le:3qbsr2c3} in Appendix \ref{sec:3qbsr2c3}, we regard the complex numbers as points on the complex plane, and discuss the following two cases.

Case (a): If the four points $1,a_4,a_3,a_3 a_4$ are not on the same line, it follows from Eq.~\eqref{eq:multi,k=1,xc} that $g+c,g a_4+c,g a_3+c,g a_3 a_4+c$ all have modulus one, and from a geometric point of view the latter four points are on the unit circle centered at the origin $(-c,0)$, where this unit circle is denoted as Circle 1. We similarly determine that the four points $g a_2,g a_2 a_4, g a_2 a_3, g a_2 a_3 a_4$ are also on Circle 1. Since $a_2$ is nonzero, we equivalently derive that $g,g a_4,g a_3,g a_3 a_4$ are on the circle of radius $\frac{1}{\abs{a_2}}$ centered at $-\frac{c}{a_2}$, where this circle is denoted as Circle 2. Thus, the four points $g,g a_4,g a_3,g a_3 a_4$ are on Circle 1 and Circle 2 simultaneously. Since we have assumed that $1,a_4,a_3,a_3 a_4$ are not on the same line, the four points $g,g a_4,g a_3,g a_3 a_4$ can only determine one circle. It implies that Circle 1 must coincide with Circle 2, and in particular the two centers of the circles have to coincide. It means $a_2=1$. However, $a_2$ cannot be $1$, otherwise $U$ would be a bipartite product unitary gate. Thus, this case is excluded.

Case (b): Suppose the four points are on the same line. On the one hand, it follows from Eq.~\eqref{eq:multi,k=1,xc} that these four points are on the same circle. On the other hand, the intersection of a line and a circle contains at most two points. Thus, there are at most two different numbers among $1,a_4,a_3,a_3 a_4$. Since $U$ is not a bipartite product unitary gate across any bipartition, it implies that $a_3\ne 1$ and $a_4\ne 1$. Then we derive that $a_3 a_4=1$, and thus $a_3=a_4=-1$. This contradicts with the assumption that $a_3\ne -1$. Thus, this case is also excluded.

Therefore, for $n\ge 4$, if $k=1$, up to a permutation of systems and under local equivalence, the Schmidt decomposition of $U$ can only be parameterized as Eq.~\eqref{eq:multi,k=1,U0}.

(v) For $k=0$, we aim to show that there is no other type of parametric Schmidt decomposition except the form in Eq.~\eqref{eq:multi,k=0,U0}. Suppose $U=A_1\otimes\cdots\otimes A_n+B_1\otimes\cdots\otimes B_n$ is the Schmidt decomposition. We first consider a special case when every $A_i$ is proportional to $I_2$, and every $B_j$ is proportional to $\sigma_3$, for any $1\leq i,j\leq n$. In this case, under local equivalence the Schmidt decomposition is parameterized as
\begin{eqnarray}
\label{eq:multi,k=0,Usimple}
U=\cos\a I_2^{\ox n}+i\sin\a\sigma_3^{\ox n},
\end{eqnarray}
where $\a\in(0,\frac{\pi}{2})$. It is a special case of Eq.~\eqref{eq:multi,k=0,U0}. Next, up to a permuation of systems we may assume that $A_1$ and $B_1$ are not simultaneously proportional to $I_2$ and $\sigma_3$, respectively.

Then, we decompose $U$ as $U=\proj{0}\ox G+\proj{1}\ox H$, where $G$ and $H$ are both unitary matrices of SR-2, and are both in the $\textrm{span}\{A_2\otimes\cdots\otimes A_n,B_2\otimes\cdots\otimes B_n\}$. It follows from Lemma \ref{le:ngeq4,k=0} (i) that $G$ and $H$ are not bipartite product matrices across any bipartition. Further, it follows from Lemma \ref{le:ngeq4,k=0} (ii) that $G$ and $H$ respectively have the following Schmidt decompositions:
\begin{eqnarray}
\label{eq:multi,gh1-1}
G&=&\diag(a,b)\ox C_3 \ox \cdots \ox C_n + \diag(c,d)\ox D_3 \ox\cdots \ox D_n,\\
H&=&\diag(p,q)\ox C_3 \ox \cdots \ox C_n + \diag(r,s)\ox D_3 \ox\cdots \ox D_n,\label{eq:multi,gh2-1}
\end{eqnarray}
where $a,b,c,d,p,q,r,s\in\mathbb{C}$, and $C_i, D_i$ ($3\le i\le n$) are all $2\times 2$ diagonal matrices. 
Let 
\begin{equation}
\label{eq:thmnqubitti}
\begin{aligned}
T_1&:=a C_3 \ox \cdots \ox C_n + c D_3 \ox\cdots\ox D_n, \quad T_2:=b C_3 \ox \cdots\ox C_n + d D_3 \ox\cdots\ox D_n, \\
T_3&:=p C_3 \ox \cdots\ox C_n + r D_3 \ox\cdots\ox D_n, \quad T_4:=q C_3 \ox \cdots\ox C_n + s D_3 \ox\cdots\ox D_n. 
\end{aligned}
\end{equation}
It follows from Eqs. \eqref{eq:multi,gh1-1} and \eqref{eq:multi,gh2-1} that $T_1, T_2, T_3, T_4$ are all unitary.
Since the four diagonal unitary matrices $T_i$'s are all linear combinations of $C_3 \ox \cdots \ox C_n$ and $D_3 \ox \cdots \ox D_n$, there are at most two linearly independent matrices among $T_1,T_2,T_3,T_4$, and some of them may be equal up to a phase. Note that $G$ and $H$ are not product unitary gates across the bipartition of the second qubit versus the last $n-2$ qubits from Lemma \ref{le:ngeq4,k=0} (i). It implies that $T_1$ and $T_2$ are not proportional to each other, and $T_3$ and $T_4$ are not proportional to each other. Finally, we discuss the following two cases, in order to show that the parametric form of $U$ is strongly constrained.

Case (a): Suppose that three of the four $T_i$'s are not proportional to each other. 
It follows from Lemma~\ref{le:spanOFdiagonal} that every $T_j$ contains at most two different diagonal elements for $1\le j\le 4$. 
Under local equivalence we may assume that $C_i=I_2$, and $D_i=\sigma_3$, for $3\le i\le n$. It follows that the Schmidt decomposition of $U$ is 
\begin{eqnarray}
\label{eq:multi,k=0}
U=A_1\ox A_2\ox I_2\ox\cdots\ox I_2+B_1\ox B_2\ox\sigma_3\ox\cdots\ox\sigma_3,
\end{eqnarray}
where $A_1,A_2,B_1,B_2$ are diagonal $2\times 2$ matrices. From $U^\dag U=I$, we obtain that
\begin{eqnarray}
\label{eq:multi,diag}
A_1^\dag A_1\ox A_2^\dag A_2 + B_1^\dag B_1\ox B_2^\dag B_2 &=& I_4,\notag\\
A_1^\dag B_1 \ox A_2^\dag B_2 + B_1^\dag A_1 \ox B_2^\dag A_2 &=& 0.
\end{eqnarray}
Assume $A_1=\diag(g_1,g_2)$, $A_2=\diag(h_1,h_2)$, $B_1=\diag(u_1,u_2)$, $B_2=\diag(v_1,v_2)$. The two equalities in Eq. \eqref{eq:multi,diag} are respectively equivalent to
\begin{equation}
\label{eq:multi,diag2}
\begin{aligned}
&\vert g_1 h_1\vert^2+\vert u_1 v_1\vert^2=\vert g_1 h_2\vert^2+\vert u_1 v_2\vert^2=\vert g_2 h_1\vert^2+\vert u_2 v_1\vert^2=\vert g_2 h_2\vert^2+\vert u_2 v_2\vert^2=1,\\
&\textrm{Re}(g_1^\ast h_1^\ast u_1 v_1)=\textrm{Re}(g_1^\ast h_2^\ast u_1 v_2)=\textrm{Re}(g_2^\ast h_1^\ast u_2 v_1)=\textrm{Re}(g_2^\ast h_2^\ast u_2 v_2)=0.
\end{aligned}
\end{equation}
According to the first continuous equality in Eq. \eqref{eq:multi,diag2}, there exist $\a,\b,\gamma,\delta\in (0,\frac{\pi}{2})$ such that
\begin{eqnarray}
\label{eq:multi,diag3}
\vert g_1 h_1\vert=\cos\a, ~\vert u_1 v_1\vert=\sin\a,&&\vert g_1 h_2\vert=\cos\b, ~\vert u_1 v_2\vert=\sin\b,\notag\\
\vert g_2 h_1\vert=\cos\gamma,~ \vert u_2 v_1\vert=\sin\gamma, &&
\vert g_2 h_2\vert=\cos\delta,~ \vert u_2 v_2\vert=\sin\delta.
\end{eqnarray}
Moreover, $\a,\b,\gamma,\delta$ satisfy $\cos\a \cos\delta=\cos\b \cos\gamma$ and $\sin\a \sin\delta=\sin\b \sin\gamma$ from Eq. \eqref{eq:multi,diag3}. They are equivalent to $\cos(\a+\delta)=\cos(\b+\gamma)$ and $\cos(\a-\delta)=\cos(\b-\gamma)$ respectively, which implies $\a=\b,~\gamma=\delta$, or $\a=\gamma,~\b=\delta$. The former implies $h_2=\pm h_1$ and $v_2=\pm v_1$, and the latter implies $g_2=\pm g_1$ and $u_2=\pm u_1$. Note that we may always assume $h_1=v_1=1$, and thus the former implies that $A_2$ and $B_2$ are $I_2$ or $\sigma_3$, and they cannot simultaneously be $I_2$ or $\sigma_3$ because $U$ is not a bipartite product unitary gate. Similarly, the latter implies that $A_1$ and $B_1$ are $I_2$ or $\sigma_3$, and they cannot simultaneously be $I_2$ or $\sigma_3$ too. By noting the phase conditions implied in Eq. \eqref{eq:multi,diag2}, we obtain the following Schmidt decomposition of $U$:
\begin{eqnarray}
\label{eq:multi,k=0,U}
U=\diag(\cos\a,\cos\b)\ox I_2^{\ox (n-1)}+i\diag(\sin\a,\sin\b)\ox\sigma_3^{\ox (n-1)},
\end{eqnarray}
where $\a,\b\neq \frac{k\pi}{2}$ for $k=0,1,2,3$, because both $\diag(\cos\a,\cos\b)$ and $\diag(\sin\a,\sin\b)$ are not singular.

Case (b): Suppose that two pairs of the $T_i$'s are proportional. We obtain either $T_1\propto T_3$ and $T_2\propto T_4$, or $T_1\propto T_4$ and $T_2\propto T_3$. In the former case we further obtain that the two scale factors $k_1,~k_2$ such that $T_3=k_1 T_2$ and $T_4=k_2 T_2$ are not equal, because $G$ and $H$ are linearly independent. Then if we expand $U$ using $C_3 \ox \cdots \ox C_n$ and $D_3 \ox \cdots \ox D_n$ on the last $n-2$ qubits, the two operators on the first two qubits respectively associated with $C_3 \ox \cdots \ox C_n$ and $D_3 \ox \cdots \ox D_n$ are not product operators, which violates that $U$ is of SR-2. Thus, this case is impossible. In the latter case we may apply some local diagonal unitary operators on the first two qubits, and then assume $T_1=T_4$, and $T_3=e^{i\phi}T_2$. Then we obtain $\diag(a,b,p,q)=\diag(a,b,e^{i\phi}b,a)$ and $\diag(c,d,r,s)=\diag(c,d,e^{i\phi}d,c)$. Since $U$ is of SR-2, we similarly conclude that the operators on the first two qubits respectively associated with $C_3 \ox \cdots \ox C_n$ and $D_3 \ox \cdots \ox D_n$ are two product operators. It follows that $b=\pm a e^{-i\phi/2}$ and $d=\pm c e^{-i\phi/2}$. Furthermore, since $U$ is not a bipartite product unitary gate across any bipartition, we have $ad=-bc$.
Hence, by further applying some local diagonal unitary gates on the first two qubits, we obtain $U=a I_2\ox I_2\ox C_3 \ox \cdots \ox C_n+c\sigma_3\ox\sigma_3\ox D_3 \ox \cdots \ox D_n$. This violates the initial assumption that $A_1$ and $B_1$ in the Schmidt decomposition are not simultaneously proportional to $I_2$ and $\sigma_3$, respectively. Hence this case is also excluded.

This completes the proof.
\end{proof}


\section{Application: The characterization of three-qubit diagonal unitary gates}
\label{sec:3qubidiagU}

In this section we further study three-qubit diagonal unitary gates under LU equivalence. As the term suggests, three-qubit diagonal unitary gates are in the form of diagonal matrix. Since diagonal unitary gates physically are controlled gates controlled from every party of the quantum system \cite[Lemma 2]{cy13}, they are indispensable for quantum circuits. For example, the two-qubit CNOT gate is a common controlled gate, which is LU equivalent to a diagonal unitary matrix. As we know from Lemma \ref{le:cohenli2013}, every unitary gate of SR-2 is LU equivalent to a diagonal one. Thus, all three-qubit unitary gates of SR-2 are contained in the set of three-qubit diagonal unitary gates under local equivalence. In Sec. \ref{subsec:twoexps}, we discuss two specific examples of three-qubit diagonal unitary gates of SR-2. It is helpful to understand the essential difference between the bipartite scenario and multipartite scenarios, and the core role of factor SN in this paper. In Sec. \ref{subsec:3qdiagU}, we show that three-qubit diagonal unitary gates have SR at most three, and give a complete characterization of genuine three-qubit diagonal unitary gates.

\subsection{Two typical examples of three-qubit diagonal unitary gates}
\label{subsec:twoexps}

In this subsection we discuss two examples of three-qubit diagonal unitaries. The first example reveals an essential difference between the tripartite scenario and the bipartite scenario, and reflects the factor SN cannot be used to classify bipartite unitary gates. The second example is the so-called CCZ gate which is very useful in quantum computation, as it is LU equivalent to the well-known Toffoli gate.

\textbf{Example 1:} Suppose $D=\diag(-1,1,1,1,1,1,1,-1)$ is a three-qubit diagonal unitary gate of system $ABC$. One can verify that $D$ is of SR-2, and it can be decomposed as
\begin{equation}
\label{eq:ex1schdec}
D=
\diag(1,-i)
\otimes 
\diag(1,-i)
\otimes
\frac{iI_2-\s_3}{2}+
\diag(1,i)
\otimes 
\diag(1,i)
\otimes
\frac{-iI_2-\s_3}{2}.
\end{equation}
Obviously, there is no local singular matrix in the Schmidt decomposition as Eq. \eqref{eq:ex1schdec}, and thus the SN of $D$ is zero. According to Theorem \ref{le:sr2=singular} (iv), the Schmidt decomposition of $D$ must coincide with the the form in Eq. \eqref{eq:sr23qubit} under local equivalence and up to a permutation of systems. It means that Eq. \eqref{eq:ex1schdec} gives a special solution for the system of equations \eqref{eq:u3qubiti}, i.e., $a=\frac{1-i}{2}$, $b=\frac{1+i}{2}$, and $c=d=-i$. Moreover, one can verify that $D$ is also LU equivalent to $I_2\ox I_2\ox \diag(\cos\a,\cos\b)+\sigma_3\ox \sigma_3\ox \diag(i\sin\a,i\sin\b)$ where $\a=\frac{\pi}{4}$ and $\b=-\frac{\pi}{4}$. It follows that the Schmidt decomposition of $D$ also coincides with the parametric form in Eq. \eqref{eq:multi,k=0,U0} from Theorem \ref{le:multi=sr2} (v). Thus, we conclude that the form in Eq. \eqref{eq:sr23qubit} is reduced to that in Eq. \eqref{eq:multi,k=0,U0} for some solutions of the system of equations \eqref{eq:u3qubiti}. Nevertheless, for example, if the solution of the system of equations \eqref{eq:u3qubiti} satisfies $\abs{c}\neq 1$ or $\abs{d}\neq 1$, the form in Eq. \eqref{eq:sr23qubit} may not be reduced to that in Eq. \eqref{eq:multi,k=0,U0}.

Furthermore, it follows directly from the SR-2 condition that either of the $AB$, $AC$ and $BC$ spaces of unitary $D$ is spanned by exactly two product matrices. From the form of $D$ in Eq. \eqref{eq:ex1schdec} , we obtain that in such a linear space spanned by two product matrices, there are at most two product (unitary) matrices up to global coefficients. Hence $D$ cannot be written as $D=\sum^2_{j=1}A_j\otimes B_j\otimes C_j$ where one of the pairs $(A_1,A_2)$, $(B_1,B_2)$ and $(C_1,C_2)$ are orthogonal projectors. Therefore, the matrix $D$ is a typical example different from the scenario of bipartite unitary gates of SR-2, as the latter allows the expression $U=P\otimes V+(I-P)\otimes W=V \otimes Q+W \otimes (I-Q)$ for some projectors $P$ and $Q$ \cite{cy13}.

\textbf{Example 2:} The so-called CCZ gate is a three-qubit diagonal unitary gate in the matrix form as $U=\diag(1,1,1,1,1,1,1,-1)$, and has the Schmidt decomposition as
\begin{equation}
\label{eq:cczsd}
U=I_2\ox I_2\ox I_2 -2\proj{1}\ox\proj{1}\ox\proj{1}.
\end{equation}
It is obvious that the CCZ gate has SR-2 and its SN is three, which reaches the upper bound of SN. Hence the CCZ gate is LU equivalent to the form given in Theorem \ref{le:sr2=singular} (i). 
Moreover, it is LU equivalent to the Toffoli gate, also known as controlled-controlled-not gate, whose matrix form is 
\begin{equation}
\label{eq:toffolim}
\begin{aligned}
U_{\text{Toffoli}}&=\proj{000}+\proj{001}+\proj{010}+\proj{011}\\
&+\proj{100}+\proj{101}+\ketbra{110}{111}+\ketbra{111}{110}.
\end{aligned}
\end{equation}
Physically, the effect of Toffoli gate is to flip the third qubit, if and only
if the first two qubits are both in state $\ket{1}$ (and does nothing otherwise), see Fig. \ref{fig:toffoli}. Quantum Toffoli gate is a fundamental three-qubit unitary gate, and has been shown to be a crucial component of many quantum information processing schemes, such as fault tolerant quantum circuits \cite{toffolift2013}, distributed quantum computation \cite{toffolidqc2017}, and quantum error correction \cite{toffoliqecc2011}. Thus, as a unitary gate LU equivalent to the Toffoli gate, the CCZ gate has also aroused great interest, and has been realized in several experimental protocols \cite{Toffoli09,Toffoli12,toffoliopt2017}. 

\begin{figure}[ht]
\centerline{
    \Qcircuit @C=8em @R=4em {
    \lstick{A} & \ctrl{1} & \rstick{A} \qw \\
    \lstick{B} & \ctrl{1} & \rstick{B} \qw \\
    \lstick{C} & \targ & \rstick{C\oplus(AB)} \qw
    }
}
\caption{The Toffoli gate: flip the third qubit conditioned on the $\ket{11}$ state of system $AB$.}
\label{fig:toffoli}
\end{figure}
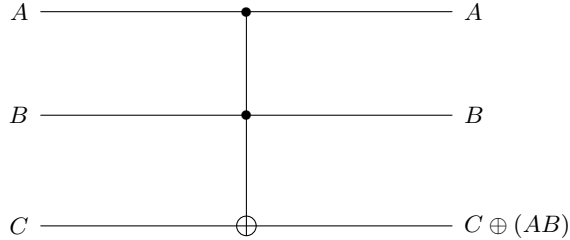


Similar to the discussion about the first example, we claim that the CCZ gate also cannot be written in the form $U=\sum^2_{j=1}A_j\otimes B_j\otimes C_j$ where one of the pairs $(A_1,A_2)$, $(B_1,B_2)$ and $(C_1,C_2)$ are orthogonal projectors. To sum up, such two examples both reflect the Schmidt decompositions for multipartite unitary gates are essentially different from that for bipartite gates.

\subsection{Characterization of three-qubit diagonal unitary gates}
\label{subsec:3qdiagU}

Every unitary gate of SR-2 is LU equivalent to a diagonal one. Conversely, whether every nonlocal genuine three-qubit diagonal unitary gate has only SR-2? If the answer is positive, the classification of genuine three-qubit unitary gates of SR-2 given by Theorem \ref{le:sr2=singular} provides a complete charcterization of genuine three-qubit diagonal unitary gates under local equivalence. Neverthless, as we shall see from the discussion below, the answer is actually negative, i.e., there exist three-qubit diagonal unitary gates whose SR is greater than two. Therefore, we further characterize three-qubit diagonal unitary gates of SR greater than two. Combined with Theorem \ref{le:sr2=singular}, we present a complete characterization of genuine three-qubit diagonal unitary gates under local equivalence.

By further study of the relationship between Eqs.~\eqref{eq:u3qubitf} and \eqref{eq:u3qubite} we confirm the existence of three-qubit diagonal unitary gates of SR greater than two. Specifically, we study whether Eqs.~\eqref{eq:u3qubitf} and \eqref{eq:u3qubite} are equivalent without considering the conditions $\delta\ne 0$ in Eq. \eqref{eq:u3qubite} and the SN $k=0$ which are only for Theorem \ref{le:sr2=singular} (iv), in order to study all possible SRs of $U$ given by Eq.~\eqref{eq:u3qubite}.

We begin with such an example: $U=\proj{0}\otimes(\cos\t I_2\otimes I_2+i\sin\t\s_3\otimes\s_3)+\proj{1}\otimes I_2\otimes \s_3$, where $\t\in(0,\frac{\pi}{2})$. It is in the form of Eq. \eqref{eq:u3qubite} under local equivalence. Via the isomorphism below
\begin{equation}
\label{eq:isomorphism}
\proj{j}_A\lra \ket{j}_A,~ (I_2)_B\lra \ket{0}_B,~ (\s_3)_B\lra \ket{1}_B,~ (I_2)_C\lra \ket{0}_C,~ (\s_3)_C\lra \ket{1}_C,
\end{equation}
we obtain a three-qubit pure state $\ket{0}(\cos\t\ket{0,0}+i\sin\t\ket{1,1})+\ket{1,0,1}$ which is isomorphic to the original unitary gate $U$. Since the SR of $U$ is the same as that of the isomorphic state, we may equivalently show the SR of this isomorphic state is greater than two. One can verify that this isomorphic state is SLOCC equivalent to the three-qubit $W$ state $\ket{\w}=\frac{1}{\sqrt3}(\ket{0,0,1}+\ket{0,1,0}+\ket{1,0,0})$. It is known that $\ket{\w}$ has SR-3 \cite{PhysRevA.81.014301}. Since the SR of a tensor is invariant under SLOCC equivalence \cite{ccd2010}, we conclude that the SR of the original unitary gate $U$ is three too. As a result, the original three-qubit diagonal unitary gate $U$ cannot be decomposed in the form as Eq. \eqref{eq:u3qubitf} even under local equivalence.

The above discussion provides an effective method to identify the SR of a three-qubit diagonal unitary gate. That is, by first mapping the three-qubit diagonal unitary gate into a pure three-qubit state via an isomorphism, then to identify the SR of the isomorphic state under SLOCC equivalence. We next use such a method to characterize the general form of three-qubit diagonal unitary gates of SR greater than two, which helps to classify the whole set of three-qubit diagonal unitary gates under local equivalence. It is well known from \cite{dvc2000} that any genuinely entangled three-qubit pure state is SLOCC equivalent to either the $GHZ$ state $\ket{\ghz}=\frac{1}{\sqrt2}(\ket{0,0,0}+\ket{1,1,1})$ or the $W$ state $\ket{\w}$.  By definition the genuinely entangled pure states are not bipartite product states across any bipartition, and they are corresponding to the genuine multipartite unitary gates via the above-mentioned isomorphism between the states and the unitary gates. Therefore, when only considering the genuine three-qubit diagonal unitary gates, there are only two SLOCC inequivalent classes. One class is related to $\ket{\ghz}$, and the other is related to $\ket{\w}$. 
Note that the term ``related'' here refers to an indirect relation via first mapping the unitary gate to some state (i.e., SLOCCa equivalence in Definition \ref{def:slocca}), rather than the direct relation by multiplying local invertible operators (i.e., SLOCC equivalence). As we know, the SR of $\ket{\ghz}$ is two and the SR of $\ket{\w}$ is three \cite{PhysRevA.81.014301}. Therefore, we shall study the class of three-qubit diagonal unitary gates that is related to $\ket{\w}$.

Similar to Eq. \eqref{eq:u3qubite}, under LU equivalence we may assume a three-qubit diagonal unitary gate as
\begin{eqnarray}
\label{eq:u3qubite0}
U=\diag(1,e^{i\a},e^{i\b},e^{i\gamma},1,1,1,e^{i\delta}),
\end{eqnarray}
where $\a,\b,\gamma,\delta\in[0,2\pi)$. The only difference from Eq.~\eqref{eq:u3qubite} is that there is no $\d\ne 0$, since this constraint is specific to the case of SR-2 and the SN $k=0$. Our task now becomes to find out when the three-qubit state isomorphic to $U$ in the form of Eq. \eqref{eq:u3qubite0}, i.e., the state 
\begin{eqnarray}
\label{eq:w0}
\ket{\psi_U}=(1,e^{i\a},e^{i\b},e^{i\gamma},1,1,1,e^{i\delta})
\end{eqnarray}
is SLOCC equivalent to $\ket{\w}$. Here, the isomorphism is given by 
\begin{equation}
\label{eq:isomorphism-2}
\proj{j,k,l}\lra \ket{j,k,l},\quad j,k,l=0,1.
\end{equation}
Since $\ket{\w}$ is genuinely entangled, the pure state $\ket{\psi_U}$ has to be genuinely entangled too. It follows that 
all of the following three matrices should have rank two.
\begin{eqnarray}
\label{eq:w1}
\bma
1 & e^{i\a} & e^{i\b} &   e^{i\g} \\
1 & 1 & 1 &   e^{i\d} \\
\ema,
\quad\quad
\bma
1 & e^{i\a} & 1 & 1  \\
e^{i\b} &   e^{i\g}& 1 &  e^{i\d} \\
\ema,
\quad\quad
\bma
1 & e^{i\b} & 1 & 1 \\
e^{i\a} &   e^{i\g} & 1 &   e^{i\d} \\
\ema.
\end{eqnarray}
Or equivalently, all of the following three conditions are met.
\begin{equation}
\label{eq:w1-1}
(e^{i\a},e^{i\b},e^{i\g})\neq (1,1,e^{i\d}),~(e^{i\b},e^{i\g},e^{i\d})\neq (1,e^{i\a},1),~(e^{i\a},e^{i\g},e^{i\d})\neq (1,e^{i\b},1).
\end{equation}
Based on such a precondition, we obtain the following result.

\begin{lemma}
\label{le:3qubitsr3}
Suppose that the four parameters $\a,\b,\g,\d\in[0,2\pi)$ satisfy the precondition given by Eq. \eqref{eq:w1-1}.
Then the three-qubit pure state $\ket{\psi_U}=(1,e^{i\a},e^{i\b},e^{i\gamma},1,1,1,e^{i\delta})$ is SLOCC equivalent to the three-qubit $W$ state $\ket{\w}=\frac{1}{\sqrt{3}}(\ket{001}+\ket{010}+\ket{100})$ if and only if one of the following two conditions holds:

(i) $\d\ne 0$, and $(e^{i\g}+e^{i\d}-e^{i\a}-e^{i\b})^2=4(e^{i\d}-1)(e^{i\g}-e^{i(\a+\b)})$;

(ii) $\d=0$, $\gamma=\pi$, and $e^{i\a}+e^{i\b}=0$ for $\a\in(0,\pi)\cup(\pi,2\pi)$.
\end{lemma}

We put the proof of Lemma \ref{le:3qubitsr3} in Appendix \ref{sec:3qubitsr3}. Note that if $\a,\b,\g,\d$ satisfy the condition in Lemma \ref{le:3qubitsr3} (ii), then they naturally satisfy the precondition given by Eq. \eqref{eq:w1-1}.

Based on Lemma \ref{le:3qubitsr3} we give a complete characterization of the genuine three-qubit diagonal unitary gates that have SR-3 as follows.

\begin{theorem}
\label{le:3qubit}
(i) Every three-qubit diagonal unitary gate is in the form of $\diag(1,e^{i\a},e^{i\b},e^{i\gamma},1,1,1,e^{i\d})$ under LU equivalence, where $\a,\b,\gamma,\delta\in[0,2\pi)$. Any such unitary gate has Schmidt rank at most three, and if it is of Schmidt rank three, it must be genuine.  

(ii) Assume that $U$ is a genuine three-qubit diagonal unitary gate in the form of $\diag(1,e^{i\a},e^{i\b},e^{i\gamma},1,1,1,e^{i\d})$ under LU equivalence. Then $U$ has Schmidt rank three if and only if the four parameters $\a,\b,\gamma,\d$ satisfy one of the two conditions in Lemma \ref{le:3qubitsr3} while making the precondition in Eq. \eqref{eq:w1-1} true.
\end{theorem}

\begin{proof}
(i) Suppose $U$ is a three-qubit diagonal unitary gate. Then we may decompose it as $U=\proj{0}\ox V+\proj{1}\ox W$, where $V$ and $W$ are both two-qubit diagonal unitary gates. One can verify that $V$ and $W$ can be transformed into $\diag(1,e^{i\a},e^{i\b},e^{i\gamma})$ and $\diag(1,1,1,e^{i\d})$ respectively after applying a proper local unitary gate $\diag(e^{i\t_{11}},1)\ox\diag(e^{i\t_{21}},e^{i\t_{22}})\ox\diag(e^{i\t_{31}},e^{i\t_{32}})$. Thus, every three-qubit diagonal unitary gate is LU equivalent to the desired form. In advantage of the isomorphism given by \eqref{eq:isomorphism-2}, we may equivalently consider the SR of the three-qubit pure state $\ket{\psi_{U}}$ in the form of Eq. \eqref{eq:w0} under SLOCC equivalence. We also conclude from the isomorphism that $U$ is a genuine three-qubit unitary gate if and only if $\ket{\psi_U}$ is genuinely entangled.
It is well known from \cite{dvc2000} that every three-qubit pure state is SLOCC equivalent to one of the following six states $1)~\ket{0,0,0},~2)~\frac{1}{\sqrt2}(\ket{0,0,0}+\ket{0,1,1}),~3)~\frac{1}{\sqrt2}(\ket{0,0,0}+\ket{1,0,1}),~4)~\frac{1}{\sqrt2}(\ket{0,0,0}+\ket{1,1,0}),~5)~\ket{\ghz}$ and $6)~\ket{\w}$, where only $\ket{\ghz}$ and $\ket{\w}$ are genuinely entangled. It is also known that such six states have SR at most three, and the upper bound is only saturated by $\ket{\w}$ \cite{PhysRevA.81.014301}. Thus, such a unitary gate has SR-3 if and only if the pure state isomorphic to this gate is SLOCC equivalent to $\ket{\w}$. It implies that such a unitary gate must be genuine if it is of SR-3. 

(ii) It follows from the above result that $U$ has SR-3 if and only if $\ket{\psi_U}$ is SLOCC equivalent to $\ket{\w}$. Thus, assertion (ii) directly follows from Lemma \ref{le:3qubitsr3}.

This completes the proof.
\end{proof}

We obtain the following corollary to characterize genuine three-qubit diagonal unitary gates of SR-2. Note that it is easy to determine whether a three-qubit diagonal unitary gate is a bipartite product unitary gate across some bipartition by checking whether the three matrices in Eq. \eqref{eq:w1} all have rank two.

\begin{corollary}
\label{cr:3qubit}
Suppose $U$ is any given nonlocal genuine three-qubit diagonal unitary gate. Assume it is in the form of Eq.~\eqref{eq:u3qubite0} under LU equivalence. If the equivalent form does not satisfy Lemma~\ref{le:3qubitsr3}, then it must be of Schmidt rank two, and its isomorphic state $\ket{\psi_U}$ must be SLOCC equivalent to $\ket{\ghz}$.
\end{corollary}

Combining Theorem~\ref{le:3qubit} and Corollary~\ref{cr:3qubit}, we give a complete characterization of genuine three-qubit diagonal unitary gates under LU equivalence. This also means we have provided the SLOCCa equivalence classes of three-qubit diagonal unitary gates by Definition \ref{def:slocca}. In particular, we can determine all the diagonal matrix forms of genuine three-qubit unitary gates of SR-2, and then use the key factor SN to classify them as Theorem \ref{le:sr2=singular}.

\section{Conclusions}
\label{sec:con}

In this paper we mainly investigated the classification of multipartite (excluding bipartite) unitary gates of Schmidt rank two (SR-2) under local equivalence. We focused on genuine multipartite unitary gates, i.e. those multipartite unitaries which are not product operators across any bipartition.
First, we proposed a key notion named as \emph{singular number} (SN) to classify the genuine multipartite unitary gates of SR-2 based on an essential observation that the Schmidt decomposition of such a unitary gate is unique. 
Then we determined all possible numbers for SN $k$. Specifically, for tripartite unitary gates, the SN $k$ can only be $0,1,2,3$, and for $n$-partite unitary gates with $n\geq 4$, the SN $k$ can only be $0,1,2,n-1,n$. 
Second, we discussed the classification of genuine multiqubit unitary gates of SR-2 using the key factor SN in detail. We divided the discussion into two parts, i.e., the part of three-qubit unitary gates and the part of $n$-qubit unitary gates with $n\geq 4$, as the ranges of the SN $k$ for such two parts are different. In each part, we formulated the parametric Schmidt decompositions of the unitary gates under LU equivalence for every SN respectively. In other words, up to a permutation of systems and under local equivalence, every genuine multiqubit unitary gate of SR-2 with SN $k$ is in the corresponding parametric form for some proper parameters. Finally, we extended our study to three-qubit diagonal unitary gates due to the close relation between diagonal unitary gates and SR-2 unitary gates. As we know, every SR-2 unitary gate is LU equivalent to a diagonal unitary matrix. Hence, we started with discussing two typical examples of SR-2, which helps us better understand the essential difference between the bipartite scenario and multipartite scenarios, and the core role of SN in the classification we have proposed.
Then we addressed the characterization of three-qubit diagonal unitary gates of SR greater than two. We have shown that the SR of a three-qubit diagonal unitary gate is at most three, and characterized the parametric form of the SR-3 diagonal unitary gates. This result completed the characterization of all genuine three-qubit diagonal unitary gates. 

Our results clearly show all equivalence classes of genuine multiqubit unitary gates of SR-2 under LU equivalence, and provide the parametric Schmidt decompositions for every SN. All these parametric forms are explicit except the case of three-qubit unitary gates with SN $k=0$, i.e., Theorem \ref{le:sr2=singular} (iv). So an interesting open problem is whether the parametric form in Theorem \ref{le:sr2=singular} (iv) can be further simplified. This is helpful for understanding the relation between Theorem \ref{le:sr2=singular} (iv) and Theorem \ref{le:multi=sr2} (v). Moreover, we believe such an essential characterization would be beneficial to introduce more controlled operations into quantum computing and quantum information processing tasks, as SR-2 unitary gates are physically regarded as controlled gates controlled from each party. Thus, it is very interesting to further explore the connections between the results in this paper and other aspects of quantum information science.

\begin{acknowledgments}
Authors were supported by the NNSF of China (Grant Nos. 11871089, 11974096), and the Fundamental Research Funds for the Central Universities (Grant No. ZG216S2005).

\end{acknowledgments}

\appendix

\section{Proof of Lemma \ref{le:ubofk}}
\label{sec:supp}

Here, we present the proof of Lemma \ref{le:ubofk}. For this purpose, we first need the following useful lemma. 

\begin{lemma}
\label{le:n=3k=1xneq1}
Suppose $\a,\b,\gamma,\delta\in(0,2\pi)$. Then the equation $(e^{i\a}-1)(e^{i\delta}-1)=(e^{i\b}-1)(e^{i\gamma}-1)$ holds if and only if either $(\a,\delta)=(\b,\gamma)$ or $(\a,\delta)=(\gamma,\b)$.
\end{lemma}

\begin{proof}
The real and imaginary parts of the above equation give the following two independent constraints.
\begin{equation}
\label{eq:angles2-02}
\left\{
\begin{aligned}
\cos(\b+\gamma)-\cos\b-\cos\gamma &=\cos(\a+\delta)-\cos\a-\cos\delta, \\
\sin(\b+\gamma)-\sin\b-\sin\gamma &=\sin(\a+\delta)-\sin\a-\sin\delta. \\
\end{aligned}
\right.
\end{equation}
By applying trigonometric formulas, Eq. \eqref{eq:angles2-02} can be simplified to
\begin{equation}
\label{eq:angles2-03}
\left\{
\begin{aligned}
\sin\frac{\b}{2}\sin\frac{\gamma}{2}\cos\frac{\b+\gamma}{2} &=\sin\frac{\a}{2}\sin\frac{\delta}{2}\cos\frac{\a+\delta}{2}, \\
\sin\frac{\b}{2}\sin\frac{\gamma}{2}\sin\frac{\b+\gamma}{2} &=\sin\frac{\a}{2}\sin\frac{\delta}{2}\sin\frac{\a+\delta}{2}. \\
\end{aligned}
\right.
\end{equation}
This system of equations is equivalent to
\begin{equation}
\label{eq:angles2-04}
\text{either}\quad
\left\{
\begin{aligned}
&\b+\gamma = \a+\delta,\\
&\sin\frac{\b}{2}\sin\frac{\gamma}{2} = \sin\frac{\a}{2}\sin\frac{\delta}{2},\\
\end{aligned}
\right.\quad
\text{or}\quad
\left\{
\begin{aligned}
&\b+\gamma = \a+\delta+2\pi,\\
&\sin\frac{\b}{2}\sin\frac{\gamma}{2} = -\sin\frac{\a}{2}\sin\frac{\delta}{2}.\\
\end{aligned}
\right.
\end{equation}
From the second equality of the former system of equations in Eq. \eqref{eq:angles2-04}, we obtain that
\begin{eqnarray}
\label{eq:h5}
&& \sin\frac{\b}{2} \sin\frac{\gamma}{2} = \sin\frac{\a}{2} \sin\frac{\delta}{2}=\sin\frac{\a}{2} \sin\frac{\b+\gamma-\a}{2}\notag\\
\Longrightarrow && \cos\frac{\b-\gamma}{2}-\cos\frac{\b+\gamma}{2}= \cos(\a-\frac{\b+\gamma}{2})-\cos\frac{\b+\gamma}{2}\notag\\
\Longrightarrow && \cos\frac{\b-\gamma}{2}= \cos(\a-\frac{\b+\gamma}{2})\notag\\
\Longrightarrow && \cos(\b-\frac{\b+\gamma}{2})= \cos(\a-\frac{\b+\gamma}{2})\notag\\
\Longrightarrow && \a=\b \quad\mbox{or}\quad \a=\gamma.
\end{eqnarray}
Furthermore, from the first equality of the former system of equations in Eq. \eqref{eq:angles2-04}, we obtain that $\gamma=\delta$ if $\a=\b$, and $\b=\delta$ if $\a=\gamma$. For the latter system of equations in Eq. \eqref{eq:angles2-04}, we similarly conclude from the second equality that $\a=\b$ or $\a=\gamma$. Then, from the first equality we obtain that $\gamma=\delta+2\pi$ if $\a=\b$, and $\b=\delta+2\pi$ if $\a=\gamma$. Since we have supposed $\a,\b,\gamma,\delta\in(0,2\pi)$ due to the periodicity, the latter in Eq. \eqref{eq:angles2-04} is not applicable here.  

This completes the proof.
\end{proof}

Note that Lemma \ref{le:n=3k=1xneq1} is also used to exclude the parameter $c\neq 1$ in Theorem \ref{le:sr2=singular} (iii). Now, we are able to present the proof of Lemma \ref{le:ubofk} as follows.

\textbf{Proof of Lemma \ref{le:ubofk}.}
Let $U=A_1\ox \cdots\ox A_n+ B_1\ox \cdots\ox B_n$ be the Schmidt decomposition.
We prove it by contradiction. Assume that there exists a genuine $n$-partite SR-2 unitary gate $U$ whose SN $k$ is $n+1$. Up to a permutation of systems, we may assume $A_1$ and $B_1$ are both singular. Then the fact that $U$ is unitary and $A_1$ is singular implies that $B_2,\cdots,B_n$ are all unitary. Simiarly, we conclude that $A_2,\cdots,A_n$ are all unitary because $B_1$ is singular. It follows that the number of local singular matrices is two only, and thus we obtain a contradiction. Next, we show the last claim that for $n\geq 5$, $k\in[3,n-2]$ is impossible.

First, if there is some $A_i$ and some $B_j$ that are both singular, then $A_s$ with $s\ne i$ are all unitary, and $B_l$ with $l\ne j$ are all unitary. It follows that $k=2$ which contradicts with $k\in[3,n]$. Thus, we may assume $B_1,\cdots,B_k$ are all singular without loss of generality. Since $U$ is of SR-2, it follows from Lemma \ref{le:cohenli2013} that $A_1,...,B_n$ are all diagonal matrices under LU equivalence. For simplicity, we take the $n$-qubit system as an example to illustrate our proof. One can similarly show the case of general multipartite systems by adding more diagonal entries into diagonal $B_{k+1},\cdots,B_n$.
When $U$ acts on the $n$-qubit system, up to a permutation of systems and under local equivalence we may further assume it as 
\begin{eqnarray}
\label{eq:multi,k=2,U}
I_2\otimes...\otimes I_2+x\proj{0}^{\otimes k}\otimes \diag(1,t_1)\otimes \cdots \otimes \diag(1,t_{n-k}),
\end{eqnarray}
where $t_1,...,t_{n-k}\in\mathbb{C}\backslash\{0,1\}$, and $x\in\mathbb{C}\backslash\{0\}$. The requirement that $t_j\ne 1, \forall j,$ follows from that $U$ is a genuine $n$-qubit unitary gate. 
From Eq.~\eqref{eq:multi,k=2,U}, it suffices to show that $k=n-2(\ge 3)$ is impossible. We prove it by contradiction. Assume $k=n-2(\ge 3)$. Then Eq.~\eqref{eq:multi,k=2,U} accurately is 
\begin{equation}
\label{eq:multi,k=2,U,reduced}
U=I_2\otimes...\otimes I_2+x\proj{0}^{\otimes n-2}\otimes \diag(1,t_1)\otimes \diag(1,t_2). 
\end{equation}
It implies that $W=\diag(1,1,1,1)+x\diag(1,t_2,t_1,t_1 t_2)$ is unitary. We may assume $W=\diag(e^{i\a},e^{i\b},e^{i\g},e^{i\d})$, where $\a,\b,\g,\d\in(0,2\pi)$ because $t_j\ne 0$, for $j=1,2$. It also requires that $\a\neq \b$ and $\a\neq \gamma$, since both $t_1,t_2$ are not equal to $1$. Then $x=e^{i\a}-1\neq 0$, and
\begin{equation}
\label{eq:multi,knot3}
\left\{
\begin{aligned}
xt_2&=e^{i\b}-1, \\
xt_1&=e^{i\g}-1, \\
xt_1 t_2&=e^{i\d}-1.
\end{aligned}
\right.
\end{equation}
It follows that $x(e^{i\d}-1)=(e^{i\b}-1)(e^{i\g}-1)$, i.e. $(e^{i\a}-1)(e^{i\d}-1)=(e^{i\b}-1)(e^{i\g}-1)$. Then, from Lemma \ref{le:n=3k=1xneq1} in Appendix \ref{sec:supp}, we conclude that the above equality holds if and only if either $(\a,\delta)=(\b,\gamma)$ or $(\a,\delta)=(\gamma,\beta)$. However, this contradicts with the restriction that $\a\neq \b$ and $\a\neq \gamma$. Hence no unitary $W$ can exist. It means that $k=n-2(\geq 3)$ is impossible. It implies that $k\in[3,n-2]$ is impossible, where $n\geq 5$.

This completes the proof.
\qed

\section{Proof of Lemma \ref{le:3qbsr2c3}}
\label{sec:3qbsr2c3}

\textbf{Proof of Lemma \ref{le:3qbsr2c3}.}
Due to Eq. \eqref{eq:3qbsr2c3-1} we may assume $f+c=e^{i\a}$, $g+c=e^{i\b}$, and $fh+c=e^{i\gamma}$ with $\a\neq \b+2k\pi$ and $\a\neq \gamma+2k\pi$ for integer $k$, and thus we obtain $h=\frac{e^{i\gamma}-c}{e^{i\a}-c}\neq 1$. Denote by $f_x,~f_y$ the real and imaginary parts of the complex number $f$ respectively, and similarly for the other two complex numbers $g,~h$. That is, $f=f_x+if_y$, $g=g_x+ig_y$, $h=h_x+ih_y$. It follows that $f_x=\cos\a-c$ and $f_y=\sin\a$. From $h=\frac{e^{i\gamma}-c}{e^{i\a}-c}$ we specifically calculate $h_x=\frac{c^2+\cos(\a-\gamma)-c(\cos\a+\cos\gamma)}{1+c^2-2c\cos\a}$ and $h_y=\frac{c(\sin\a-\sin\gamma)-\sin(\a-\gamma)}{1+c^2-2c\cos\a}$. Then it remains to determine the complex number $g$, or equivalently the phase $\b$, by the two parameters $\a,~\gamma$.
To figure out all the parametric expressions of $g$, we regard complex numbers as points on the complex plane where the X-axis represents the real part and the Y-axis represents the imaginary part. To better describe our explanation, we mark $\overline{fh,gh}$ as the line segment with two endpoints $fh$ and $gh$ on the complex plane. By direct calculation we obtain the coordinates of points $fh$ and $gh$ as 
\begin{equation}
\label{eq:3qubitcoord-2}
\begin{aligned}
\mathrm{Re}(fh)&=f_xh_x-f_yh_y,\quad \mathrm{Im}(fh)=f_xh_y+f_yh_x,\\
\mathrm{Re}(gh)&=g_xh_x-g_yh_y,\quad \mathrm{Im}(gh)=g_xh_y+g_yh_x.
\end{aligned}
\end{equation} 
Since the condition $\abs{f+c}=\abs{g+c}=1$ from Eq. \eqref{eq:3qbsr2c3-1} is equivalent to $\abs{fh+ch}=\abs{gh+ch}=\abs{h}$, from a geometric point of view it implies that the point $(-ch_x,-ch_y)$ is in the perpendicular bisector of $\overline{fh,gh}$. Recall that $\abs{fh+c}=\abs{gh+c}$ from Eq. \eqref{eq:3qbsr2c3-1}, so we similarly determine that the point $(-c,0)$ is also in the perpendicular bisector of $\overline{fh,gh}$.
Thus, the slope of the perpendicular bisector of $\overline{fh,gh}$ is $\frac{h_y}{h_x-1}$. If $h_y=0$, then $\overline{fh,gh}$ is perpendicular to the X-axis. If $h_x=1$, then $\overline{fh,gh}$ is parallel to the X-axis. Note that $h\neq 1$. We shall consider such two cases: Case (i) $h_y=0$ and Case (ii) $h_y\neq 0$. It follows from the expression of $h_y$ derived above that $h_y=0$ if and only if $c(\sin\a-\sin\gamma)-\sin(\a-\gamma)=0$. One can verify that $c(\sin\a-\sin\gamma)-\sin(\a-\gamma)=2\sin\frac{\a-\gamma}{2}\big(c\cos\frac{\a+\gamma}{2}-\cos\frac{\a-\gamma}{2}\big)$. Due to $\a\neq \gamma$ we further conclude that $h_y=0$ if and only if $c\cos\frac{\a+\gamma}{2}-\cos\frac{\a-\gamma}{2}=0$.

Case (i). If $h_y=0$, it means $\gamma$ is dependent on $\a$ for $c\cos\frac{\a+\gamma}{2}-\cos\frac{\a-\gamma}{2}=0$, and thus there is only one free parameter $\a$ in this case. Since $fh=f_xh_x+i(f_yh_x)$ and $\abs{fh+c}=1$, we obtain 
\begin{equation}
\label{eq:3qubitcb-1}
(c+f_xh_x)^2+(f_yh_x)^2=1.
\end{equation}
Substituting $f_x=-c+\cos\a$ and $f_y=\sin\a$ into Eq. \eqref{eq:3qubitcb-1} it follows that $(1+c^2-2c\cos\a)h_x^2+2c(-c+\cos\a)h_x+c^2-1=0$. The two roots for $h_x$ are $1$ and $\frac{c^2-1}{1+c^2-2c\cos\a}$. Since $h\neq 1$, we conclude that $h_x=\frac{c^2-1}{1+c^2-2c\cos\a}$ with $c\cos\a\neq 1$.
Furthermore, we obtain that $\mathrm{Re}(fh)=\mathrm{Re}(gh)$ and $\mathrm{Im}(fh)+\mathrm{Im}(gh)=0$, since the point $(-c,0)$ is in the perpendicular bisector of $\overline{fh,gh}$. It follows from Eq. \eqref{eq:3qubitcoord-2} that $g_x=f_x$ and $g_y=-f_y$. To satisfy the constraint $f\neq g$, we have to restrict $f_y=\sin\a\neq 0$. Thus, we have formulated the analytic expressions of $f,~g,~h$ with a parameter $\a\in(0,\pi)\cup(\pi,2\pi)$ satisfying that $c\cos\a\neq 1$ for a given positive number $c\neq 1$.

Case (ii). If $h_y\neq 0$, it is equivalent to $c\cos\frac{\a+\gamma}{2}-\cos\frac{\a-\gamma}{2}\neq 0$. It follows that the slope of $\overline{fh,gh}$ is $\frac{1-h_x}{h_y}$. By Eq. \eqref{eq:3qubitcoord-2} we obtain the following equation
\begin{equation}
\label{eq:3qubitcc-1}
\frac{(g_y-f_y)h_x+(g_x-f_x)h_y}{(g_x-f_x)h_x+(f_y-g_y)h_y}=\frac{1-h_x}{h_y}.
\end{equation}
After simplification, the above equation is equivalent to 
\begin{equation}
\label{eq:3qubitcc-2}
(g_x-f_x)(h_x^2+h_y^2)=(g_x-f_x)h_x-(g_y-f_y)h_y.
\end{equation}
If $g_x-f_x=0$, it follows directly from Eq. \eqref{eq:3qubitcc-2} that $g_y-f_y=0$, which implies that $f=g$. However, it contradicts with the condition $f\neq g$. So we conclude $g_x-f_x\neq 0$, and thus $h_x^2+h_y^2=h_x-\frac{g_y-f_y}{g_x-f_x}h_y$. One can verify that $h_x^2+h_y^2=\abs{h}^2=\frac{1+c^2-2c\cos\gamma}{1+c^2-2c\cos\a}$. Hence, the equality $h_x^2+h_y^2=h_x-\frac{g_y-f_y}{g_x-f_x}h_y$ is equivalent to 
\begin{equation}
\label{eq:3qubitcc-5}
1=\cos(\a-\gamma)-c(\cos\a-\cos\gamma)-\frac{\sin\b-\sin\a}{\cos\b-\cos\a}\big(c(\sin\a-\sin\gamma)-\sin(\a-\gamma)\big).
\end{equation}
Since $h_y\neq 0$, we obtain that
\begin{equation}
\label{eq:3qubitcc-6}
\begin{aligned}
\frac{\sin\b-\sin\a}{\cos\b-\cos\a}&=\frac{-1+\cos(\a-\gamma)-c(\cos\a-\cos\gamma)}{c(\sin\a-\sin\gamma)-\sin(\a-\gamma)}\\
&=\frac{c\sin\frac{\a+\gamma}{2}-\sin\frac{\a-\gamma}{2}}{c\cos\frac{\a+\gamma}{2}-\cos\frac{\a-\gamma}{2}}, \quad \text{for $\sin\frac{\a-\gamma}{2}\neq 0$.}
\end{aligned}
\end{equation}
Furthermore, one can verify $\frac{\sin\b-\sin\a}{\cos\b-\cos\a}=-\cot\frac{\a+\b}{2}$ when $\a\neq \b+2k\pi$. Let 
\begin{equation}
\label{eq:3qubitcc-7}
F_c(\a,\gamma):=\frac{c\sin\frac{\a+\gamma}{2}-\sin\frac{\a-\gamma}{2}}{c\cos\frac{\a+\gamma}{2}-\cos\frac{\a-\gamma}{2}}.
\end{equation}
Then we conclude that $\b=(2k+1)\pi+2\arctan(F_c(\a,\gamma))-\a$, where $k$ is an integer and $\arctan(F_c(\a,\gamma))\in(-\frac{\pi}{2},\frac{\pi}{2})$. Since $\a\neq \b+2k\pi$, we obtain the constraint that $\a\neq\frac{\pi}{2}+\arctan(F_c(\a,\gamma))+k\pi$ for integer $k$. Thus, we have formulated the analytic expressions of $f,~g,~h$ with the two parameters $\a,~\gamma$ satisfying all the constraints in this case. Specifically, due to the periodicity we may assume $\a,\g\in[0,2\pi)$, and such two parameters satisfy $\a\neq \gamma$, $\a\neq\frac{\pi}{2}+\arctan(F_c(\a,\gamma))$, $\a\neq\frac{3\pi}{2}+\arctan(F_c(\a,\gamma))$, and $c\cos\frac{\a+\gamma}{2}-\cos\frac{\a-\gamma}{2}\neq 0$.

This completes the proof.
\qed

\section{Proof of Lemma \ref{le:ngeq4,k=0}}
\label{sec:ngeq4,k=0}

\textbf{Proof of Lemma \ref{le:ngeq4,k=0}.} Suppose the genuine $n$-qubit unitary gate $U=\proj{0}\ox G+\proj{1}\ox H$ has the Schmidt decomposition as $U=A_1\otimes\cdots\otimes A_n+B_1\otimes\cdots\otimes B_n$. Then we obtain that both $G$ and $H$ are in the $\textrm{span}\{A_2\otimes\cdots\otimes A_n,B_2\otimes\cdots\otimes B_n\}$.

(i) Since $U$ is a unitary gate of SR-2, it follows that $G$ and $H$ are both unitary, and have SR at most two. Furthermore, it follows from the SN $k=0$ that both $G$ and $H$ have SR-2, otherwise the SN $k>0$. 
Next, we prove the assertion that $G$ and $H$ are both genuine $(n-1)$-qubit unitary gates. Assume $G$ is a bipartite product matrix across some bipartition, i.e., $G=G_1\ox G_2$, where $G_1$ acts on a true subset $S_1$ of the $n-1$ qubits, and $G_2$ acts on the subset $S_2$ consisting of the remaining qubits. Since $G$ has SR-2, without loss of generality we may assume $G_1$ has SR-2, and $G_2$ is a product unitary gate on the $\abs{S_2}$ qubits. On the one hand, we conclude that $G_2$ is equal to either $\bigotimes_{j\in S_2} A_j$ or $\bigotimes_{j\in S_2} B_j$ up to a constant factor, as $G\in\textrm{span}\{A_2\otimes\cdots\otimes A_n,B_2\otimes\cdots\otimes B_n\}$. On the other hand, since $G_1$ is of SR-2, the Schmidt decomposition of $G_1$ must be $G_1=x_1\bigotimes_{j\in S_1} A_j+x_2\bigotimes_{j\in S_1} B_j$ for $x_1x_2\neq 0$. Recall that $U$ is a genuine $n$-qubit unitary gate, which means $A_i$ and $B_i$ are linearly independent from each other for any $1\leq i\leq n$. Thus, we derive a contradiction that $G$ is not in the $\textrm{span}\{A_2\otimes\cdots\otimes A_n,B_2\otimes\cdots\otimes B_n\}$. Therefore, $G$ is not a bipartite product matrix across any bipartition. Similarly, we obtain the same assertion for $H$.

(ii) It follows from assertion (i) that $G$ and $H$ are both diagonal unitary gates under local equivalence. Then we may assume
\begin{eqnarray}
\label{eq:multi,gh1}
G&=&\diag(a,b)\ox C_3 \ox \cdots \ox C_n + \diag(c,d)\ox D_3 \ox\cdots \ox D_n,\\
H&=&\diag(p,q)\ox E_3 \ox \cdots \ox E_n + \diag(r,s)\ox F_3 \ox\cdots \ox F_n,\label{eq:multi,gh2}
\end{eqnarray}
where $a,b,c,d,p,q,r,s\in\mathbb{C}$, and $C_i, D_i, E_i, F_i$ ($3\le i\le n$) are all $2\times 2$ diagonal matrices. 
By applying Lemma~\ref{le:linindep} (ii) to the $(n-1)$-qubit unitary gate $G$, there must be no other linear combination of $C_3 \ox \cdots \ox C_n$ and $D_3 \ox \cdots \ox D_n$ to expand $G$ except the form as Eq. \eqref{eq:multi,gh1}. It means that the Schmidt decomposition of $G$ is unique. Since $U$ has SR-2 across the bipartition of the first two qubits versus the other qubits, it implies that the four operators $C_3 \ox \cdots \ox C_n$, $D_3 \ox \cdots \ox D_n$, $E_3 \ox \cdots \ox E_n$ and $F_3 \ox \cdots \ox F_n$ are all in the two-dimensional operator space $\lin\{A_3\ox\cdots\ox A_n,B_3\ox\cdots\ox B_n\}$. Furthermore, since $C_3 \ox \cdots \ox C_n$ and $D_3 \ox \cdots \ox D_n$ are linearly independent from each other, they span the two-dimensional space $\lin\{A_3\ox\cdots\ox A_n,B_3\ox\cdots\ox B_n\}$. Similarly, since $E_3 \ox \cdots \ox E_n$ and $F_3 \ox \cdots \ox F_n$ are also linearly independent from each other, they span the same two-dimensional space. Specifically, that is
\begin{equation}
\label{eq:ngeq4c5space-1}
\begin{aligned}
\lin\{A_3\ox\cdots\ox A_n, B_3\ox\cdots\ox B_n\}&=\lin\{C_3 \ox \cdots \ox C_n, D_3 \ox \cdots \ox D_n\}\\
&=\lin\{E_3 \ox \cdots \ox E_n, F_3 \ox \cdots \ox F_n\}.
\end{aligned}
\end{equation}
In other words, both $E_3 \ox \cdots \ox E_n$ and $F_3 \ox \cdots \ox F_n$ are in the $\lin\{C_3 \ox \cdots \ox C_n,D_3 \ox \cdots \ox D_n\}$. It follows that $H=\diag(p',q')\ox C_3 \ox \cdots \ox C_n + \diag(r',s')\ox D_3 \ox\cdots \ox D_n$, where $p',q',r',s'\in\mathbb{C}$. 

This completes the proof.
\qed

\section{Proof of Lemma \ref{le:3qubitsr3}}
\label{sec:3qubitsr3}

\textbf{Proof of Lemma \ref{le:3qubitsr3}.}
Under the precondition in Eq. \eqref{eq:w1-1}, it follows from the paragraph above \cite[Eq. (18)]{dvc2000} that a genuinely entangled three-qubit pure state is SLOCC equivalent to $\ket{\w}$ if and only if the range of its bipartite marginal of system $BC$ has exactly one product vector. We discuss the following two cases.

Case (i). Suppose that $(1,1,1,e^{i\d})$ is not a product vector, i.e., $\d\neq 0$. Then the condition above can be stated as: there is exactly one solution $x\in\bbC$ such that $(1,e^{i\a},e^{i\b},e^{i\gamma})+x(1,1,1,e^{i\d})$ is a product vector. That is, the equation $(x+1)(e^{i\d}x+e^{i\g})=(x+e^{i\a})(x+e^{i\b})$
has exactly one solution $x$. Equivalently, the equation
\begin{eqnarray}
\label{eq:w3}
(e^{i\d}-1)x^2+(e^{i\g}+e^{i\d}-e^{i\a}-e^{i\b})x+e^{i\g}-e^{i(\a+\b)}=0
\end{eqnarray}
has exactly one solution.
Since we have supposed $\d\neq 0$ in this case, Eq.~\eqref{eq:w3} has exactly one solution $x$ if and only if the discriminant is equal to zero, i.e., $(e^{i\g}+e^{i\d}-e^{i\a}-e^{i\b})^2=4(e^{i\d}-1)(e^{i\g}-e^{i(\a+\b)})$. This case gives the condition (i) of this lemma.

Case (ii). Suppose that $(1,1,1,e^{i\d})$ is a product vector, i.e., $\d=0$. Then $(1,e^{i\a},e^{i\b},e^{i\gamma})$ must not be a product vector, otherwise the SR of $\ket{\psi_U}$ is at most two. It follows that $e^{i\gamma}\neq e^{i(\a+\b)}$. Then the initial necessary and sufficient condition can be similarly stated as: there is only one solution $y=0$ such that $y(1,e^{i\a},e^{i\b},e^{i\gamma})+(1,1,1,1)$ is a product vector. That is, the equation $(y+1)(e^{i\gamma}y+1)=(e^{i\a}y+1)(e^{i\b}y+1)$ has exactly one solution $y=0$. Equivalently, the equation 
\begin{equation}
\label{eq:w3-2}
(e^{i\gamma}-e^{i(\a+\b)})y^2+(1+e^{i\gamma}-e^{i\a}-e^{i\b})y=0
\end{equation}
has exactly one solution $y=0$. Since we have derived that $e^{i\gamma}\neq e^{i(\a+\b)}$ in this case, Eq. \eqref{eq:w3-2} has exactly one solution $y=0$ if and only if $1+e^{i\gamma}-e^{i\a}-e^{i\b}=0$. It follows that $e^{i\gamma}=e^{i\a}+e^{i\b}-1$ which implies that $\abs{e^{i\a}+e^{i\b}-1}=1$, i.e., $(\cos\a+\cos\b -1)^2+(\sin\a+\sin\b)^2=1$. Then one can deduce from the sum-to-product identity: $\cos\a+\cos\b=2\cos\frac{\a+\b}{2}\cos\frac{\a-\b}{2}$ that
\begin{equation}
\label{eq:w3-c2-1}
\begin{aligned}
&\qquad~ (\cos\a+\cos\b -1)^2+(\sin\a+\sin\b)^2=1 \\
&\Longleftrightarrow (\cos\a+\cos\b)^2-2(\cos\a+\cos\b)+(\sin\a+\sin\b)^2=0 \\
&\Longleftrightarrow 1+\cos(\a-\b)-(\cos\a+\cos\b)=0 \\
&\Longleftrightarrow 2\cos^2\frac{\a-\b}{2}-2\cos\frac{\a+\b}{2}\cos\frac{\a-\b}{2}=0 \\
&\Longleftrightarrow \cos\frac{\a-\b}{2}\big(\cos\frac{\a-\b}{2}-\cos\frac{\a+\b}{2}\big)=0.
\end{aligned}
\end{equation}
Moreover, we have to exclude the possibility that $e^{i\gamma}=e^{i(\a+\b)}$. From $e^{i\gamma}=e^{i\a}+e^{i\b}-1$ we may equivalently consider when the two equalities: $\cos\a+\cos\b-1=\cos(\a+\b)$ and $\sin\a+\sin\b=\sin(\a+\b)$ hold simultaneously. One can deduce from another sum-to-product identity: $\sin\a+\sin\b=2\sin\frac{\a+\b}{2}\cos\frac{\a-\b}{2}$ that
\begin{equation}
\label{eq:w3-c2-2}
\begin{aligned}
&\qquad~ \cos\a+\cos\b-1=\cos(\a+\b) \\
&\Longleftrightarrow 2\cos\frac{\a+\b}{2}\cos\frac{\a-\b}{2}=2\cos^2\frac{\a+\b}{2} \\
&\Longleftrightarrow \cos\frac{\a+\b}{2}\big(\cos\frac{\a-\b}{2}-\cos\frac{\a+\b}{2}\big)=0, \\
&\qquad~ \sin\a+\sin\b=\sin(\a+\b) \\
&\Longleftrightarrow 2\sin\frac{\a+\b}{2}\cos\frac{\a-\b}{2}=2\sin\frac{\a+\b}{2}\cos\frac{\a+\b}{2} \\
&\Longleftrightarrow \sin\frac{\a+\b}{2}\big(\cos\frac{\a-\b}{2}-\cos\frac{\a+\b}{2}\big)=0.
\end{aligned}
\end{equation}
Since $\sin\frac{\a+\b}{2}$ and $\cos\frac{\a+\b}{2}$ cannot be zero simultaneously, it follows that $e^{i\gamma}=e^{i(\a+\b)}$ if and only if $\cos\frac{\a-\b}{2}-\cos\frac{\a+\b}{2}=0$. Hence, we conclude that $\cos\frac{\a-\b}{2}-\cos\frac{\a+\b}{2}\neq 0$, i.e., $\a,\b\in(0,2\pi)$, and thus $\cos\frac{\a-\b}{2}=0$ from Eq. \eqref{eq:w3-c2-1}, i.e, $\a-\b=\pm\pi$. It is equivalent to $e^{i\a}+e^{i\b}=0$, and thus $e^{i\gamma}=-1$ from $e^{i\gamma}=e^{i\a}+e^{i\b}-1$. Since $e^{i\gamma}\neq e^{i(\a+\b)}$, it follows that $-e^{i2\a}\neq -1$, and thus $e^{i\a}\neq \pm 1$. This case gives the condition (ii) of this lemma.

This completes the proof.
\qed

\bibliographystyle{unsrt}

\bibliography{multiUNITARYsr2}

\end{document}